%% file: main.tex
\documentclass[a4paper,12pt]{article}
\usepackage{biblatex}
\usepackage{graphicx} 
\usepackage[margin=1in]{geometry}
\usepackage{booktabs}
\usepackage{mathtools}
\usepackage{amsmath}
\usepackage{amsthm}
\usepackage{amssymb}
\usepackage{amsfonts}
\usepackage{bbm}
\usepackage{graphicx}
\usepackage{titling}
\usepackage{multicol}
\usepackage[font=small]{caption}
\usepackage[capposition=top]{floatrow}
\usepackage{enumitem}
\newcommand{\qcite}[1]{\citeauthor{#1} (\citeyear{#1})}
\newcommand{\qcitenp}[1]{\citeauthor{#1} \citeyear{#1}}

\usepackage{tikz}
\usetikzlibrary{arrows}
\usetikzlibrary{positioning}
\usepackage{pgfplots}

\newcommand{\succsimq}{\mathrel{\raisebox{-0.05em}?\kern-0.3em\succsim}}
\newcommand{\eqtrir}{%
  \mathrel{%
    \ooalign{%
      $\vartriangleright$\cr
      \hidewidth\rule[ -0.5ex ]{1.8ex}{0.12ex}\hidewidth\cr
    }%
  }%
}
\newcommand{\trir}{\vartriangleright}

\newcommand{\prightarrow}{\xrightarrow{\hspace{5pt}p\hspace{5pt}}}
\newcommand{\ninfrightarrow}{\xrightarrow{n \rightarrow \infty}}

\newtheorem{theorem}{Theorem}

\theoremstyle{definition}
\newtheorem{definition}{Definition}
\newtheorem{lemma}{Lemma}
\newtheorem{corollary}{Corollary}
\newtheorem{assumption}{Assumption}

\newtheorem{example}{Example}

\title{How to Use Prices for Efficient Online Matching}
\author{Terence Highsmith II$^1$}
\date{Version: \today}

\begin{document}

\maketitle
\begin{center}
    \begin{minipage}{0.9\textwidth}
\textit{Many matching markets feature unknown, dynamic arrivals of agents that must match immediately. A caseworker must match an abused child to a foster home, a hospital must assign a patient in critical condition to a room, or a city must place a homeless individual into a shelter. We design an online matching algorithm---the Sequential Equilibrium Mechanism (SEM)---that approximates large market equilibria to match arriving agents to objects. SEM is asymptotically efficient, fair, and strategy-proof with probability one. Our application plans to deploy a lab-in-the-field experiment where real caseworkers match vulnerable children to host homes, and we provide simulation evidence that SEM can substantially improve welfare.} (JEL C78, D47, D71)	
\end{minipage}
\end{center}

\vspace{1em}

\input{Paper/Introduction}

\input{Paper/Preliminaries}

\input{Paper/Results}

\printbibliography

\appendix
\counterwithin{table}{section}
\counterwithin{figure}{section}
\counterwithin{theorem}{section}
\counterwithin{proposition}{section}
\counterwithin{lemma}{section}

\input{Paper/Appendix}

\input{Paper/Appendix-B}

\end{document}

%% file: Paper/Introduction.tex
\section*{1 \hspace{5pt} Introduction}

Many matching markets feature stochastic, dynamic ("online") arrivals of agents or objects that a market designer (e.g. firm or government) must allocate. A caseworker must match an abused child to a foster home, a hospital must assign a patient in critical condition to a room, or a city must place a homeless individual into a shelter. Unlike traditional online matching problems where the organization can discard arrivals, the settings referenced above feature arrivals that \textit{must} be matched immediately upon arrival whenever possible. This constraint is known as \textit{greedy allocation}. Greedy allocation problems may appear to be trivial: if an agent arrives, she must be matched, and allocating her most preferred object to her is constrained-efficient within the set of greedy allocations. In reality, agents often have indifferences in their preferences: two foster homes might be equally suitable for a child, two heterogeneous hospital rooms might be equally adequate for patient care, or two shelters might provide similar services to a homeless individual. Moreover, these preferences are often ordinal in lieu of a clear numeric objective. In these cases, randomly allocating one of the agent's most preferred objects to her can unnecessarily jeopardize the welfare of future arrivals.

Our online matching problem faces two core constraints: (1) greedy allocation and (2) coarse, \textit{ordinal} preferences. Agents reveal their preferences upon arrival, and the goal is to find a matching that is greedy, fair (subject to the constraints), and ex-post efficient. Surprisingly, there is no online matching mechanism that satisfies any version of Pareto efficiency for ordinal preferences. We design a novel, randomized online matching mechanism that satisfies these constraints: the \textit{Sequential Equilibrium Mechanism} (SEM). Our design is in two stages: (i) we solve a large market competitive equilibrium using a new equilibrium concept. We allot fake, token money to agents and solve for a competitive equilibrium that uses random prices to smooth agent demand. In stage (ii), SEM opens \textit{spot markets} in each period. The spot markets solve a large market equilibrium for the realized arrivals and expected future demand, and we allocate objects according to the equilibrium allocation.

SEM is a randomized mechanism that allocates lotteries to arrivals and immediately realizes the lottery at the end of each time period. We refer to the realized lotteries as ex-post lottieres and the realized allocations as ex-post allocations. Under careful use of token money, SEM is greedy. It ensures that arriving agents are allocated lotteries containing their most preferred objects whenever possible. Equivalently, every ex-post allocation within a time period is efficient with respect to the current arrivals and object supply. Moreover, SEM's lotteries are equal-type envy-free. Agents arriving at the same time period do not envy each other's lotteries. SEM is \textit{strategyproof with probability one}: an agent's strategic incentives to manipulate vanish as the number of arrivals per period grows large. Finally, our main result is that SEM is \textit{asymptotically efficient}. SEM's ex-post lottery is ordinally efficient with probability one as the number of arrivals per period grows. In contrast, it is impossible to design a greedy mechanism that achieves ex-post ordinal efficiency when the market is finite. 


Our results for SEM are very encouraging, but they are not definitive. That is, although SEM is asymptotically efficient, its finite sample performance is unknown. To test SEM's empirical effects, we apply our theoretical results in the context of a U.S.-based nonprofit (henceforth, the Firm) that places children into temporary homes. The Firm is a national organization that contracts with state child welfare departments to find temporary, emergency homes for children that are at risk of abuse or neglect. In addition, the Firm offers pro-bono services to parents that voluntarily request emergency hosting for their children. In the Firm's problem, children stochastically arrive over time and must be matched to a host home. Their current matching process is rudimentary: upon an child's arrival, every host home receives a text message about the child, and the first home to respond can host the child. The Firm's objective is to efficiently allocate children to homes, but this process was unable to accomplish their goal because even homes that are perfectly aligned with the Firm's objective can fail to coordinate on efficient allocations.

We present simulation results that compare Serial Dictatorship with Random Tie-Breaking (SD-RTB) to SEM. Upon arrival of children, SD-RTB selects a random priority order over the arrivals. In order, SD-RTB allocates to a child a random home among her most preferred homes. Since we have not yet obtained data on observed matches from the Firm, we use SD-RTB as a simplification to represent their existing matching process and simulate a simple market to compare the mechanisms. Thus, our simulations only compare two theoretical mechanisms. They do not provide counterfactuals relative to the ground truth. Our ongoing work with the Firm hopes to implement an experimental approach to assess the empirical effects of SEM. Our simulations analyzes the number of matches while varying the number of arrivals to assess SEM's rate of convergence to efficiency.

The simulations show that SEM delivers approximately 10\% more matches than SD-RTB, a finding that is steady regardless of the market size. In line with our theoretical predicts, both mechanisms' performance experience less variance as the arrival rate increases, demonstrating that aggregate uncertainty smooths out in large markets. Overall, our simulations suggest that SEM can substantially improve welfare in practice.

The success of SEM relies on a careful solution to the large market problem. Pseudomarkets, originating with \qcite{hylland-1979}, constitute the foundational approach to allocate indivisible objects to agents using competitive equilibria. In our case, agent preferences are ordinal with indifferences, making it technically challenging to prove equilibrium existence. An agent's demand is a correspondence when agents have indifferences. Particularly, it is an ill-behaved correspondence that is fully discontinuous. However, restricting to prices that are not exactly equal to the budget constraint, the demand correspondence is upper hemicontinuous, suggesting that "smoothing" over such prices could create upper hemicontinuous demand and allow the market designer to find prices that clear the market: $\mathcal{D} = \mathcal{S}$.

Briefly, our technique is to solve a large-market allocation problem with \textit{random prices}. Our equilibrium concept then focuses on setting \textit{expected demand} equal to supply. Expected demand is the Aumann integral: integration over selections from the demand correspondence. Expected demand is convex-valued and upper hemicontinuous. We leverage this to apply Kakutani's Fixed Point Theorem ala traditional existence arguments. Combining this with asymmetric budgets--larger for earlier arrivals---guarantees that earlier agents have greater purchasing power and always receive their most preferred objects when possible, satisfying the greedy allocation constraint. Simultaneously, the expected demand of future arrivals imputes the value of objects into the price vector and allows it to guide current allocations efficiently. 

Random prices also allow us to implement \textit{random tie-breaking}. We show how a naive approach to do so fails, but correctly specifying the price error term allows the market designer to find greedy, envy-free, and ordinally efficient allocations under a condition that can be verified ex-post in order to simplify computational complexity.

\subsection*{Literature Contribution}

Our work contributes to three separate, although sometimes interlocking, bodies of literature: pseudomarkets, online matching, and large market analysis.

\textbf{Pseudomarkets.} \qcite{hylland-1979} first proposed utilizing competitive markets to allocate indivisible objects (jobs) to agents (employees). In this \textit{pseudomarket} approach, agents receive token money to purchase probability shares of receiving an object; the market designer searches for a price vector that clears the market without transfers. Unlike the arguably most common class of allocation mechanisms, Serial Dictatorships, pseudomarkets are ex-ante efficient and ex-ante fair (envy-free) when the designer allots equal token money to all agents. \qcite{miralles2021foundations} proves that the First and Second Welfare Theorems from general equilibrium theory extend to pseudomarkets; \qcite{pycia-2023} offers a comprehensive overview of the standard model.

Recently, many economists have followed up on \qcite{hylland-1979} to solve more complex allocation problems. \qcite{budish2011combinatorial} uses competitive equilibrium that approximately clear the market to efficiently allocate objects---course seats, in application---to agents with combinatorial demand. Further, unlike \citeauthor{hylland-1979}'s original proposal for agents to purchase probability shares, allocations are deterministic in \citeauthor{budish2011combinatorial}'s framework. \citeauthor{budish2011combinatorial}'s approximate competitive equilibrium from equal incomes (ACEEI) mechanism is approximately efficient, envy-bounded by one object, and strategyproof in the large. \qcite{kornbluth2021undergraduate} extend \qcite{budish2011combinatorial} for the case when courses have priorities over students, as is common in undergraduate course allocation. Their mechanism satisfies the same properties, modulo priorities, and is also approximately stable. 

\qcite{he2018pseudo} focuses on the \textit{school choice} setting where agents have unit demand and schools ("objects") have priorities over students. They construct a random mechanism that is ex-ante stable, fair, constrained efficient, and ex-ante two-sided efficient. \qcite{echenique2021}, \qcite{nguyen2021}, and \qcite{gul2024efficient} develop frameworks for more complex settings such as allocation problems with floor or ceiling constraints.

\qcite{nguyen2025efficiency} design a pseudomarket mechanism for combinatorial allocation with ordinal preferences that characterizes the set of ordinally efficient allocations. Their equilibrium concept, like ours, uses random perturbations of demand to guarantee existence. Their setting assumes that agents have strict preferences, whereas one of our key contributions is to extend their approach to allow for indifferences. Non-trivial technical problems arise that complicate applying fixed point theorems to show equilibrium existence. As a result, our proofs differ from \qcite{nguyen2025efficiency}. In addition: while they introduce random budgets, we introduce random prices. While we conjecture that these two approaches are equivalent in the strict preference domain, they are not on a general domain. Different assumptions on the price error term correspond to different modes of tie-breaking (or lack thereof). Comparatively, random budgets are less flexible. Further, we show a link between the validity of tie-breaking and asymptotic efficiency. Last, they establish results for combinatorial allocation; our work applies our equilibrium concept to design a mechanism that satisfies desirable properties for online matching.

Related, \qcite{bogomolnaia2001probabilistic} introduce Probabilistic Serial (PS) in assignment problems with unit demand, motivating stochastic dominance as a welfare criterion. Although, their solution is not a pseudomarket mechanism. We highlight \qcite{bogomolnaia2001probabilistic} because, along with follow-ups that extend PS (EPS) to environments with indifference, EPS solves our static problem (\qcitenp{katta2006solution}). However, it is unclear if it is possible to reconfigure EPS for online matching problems; this motivates our use of pseudomarkets and is our contribution beyond the static problem. Moreover, we are able to show a condition wherein a form of random tie-breaking is without loss of efficiency.

\textbf{Online Matching.} \qcite{karp1990optimal} introduces the canonical model for online matching: agents or objects arrive sequentially over time, matching decisions are irrevocable, and arrivals depart immediately if unmatched. In the original problem, the goal is to maximize the number of matches; generally, the goal is to maximize a sum of cardinal match valuations. The ratio between algorithm's worst-case ("adversarial") performance and optimal performance is known as the competitive ratio. The celebrated Ranking algorithm achieves a competitive ratio of $1 - 1/e$. Our present work assumes ordinal preferences and develops a mechanism with optimal asymptotic performance rather than an adversarial guarantee.

\qcite{feldman2014online} study online matching under known arrival distributions and show that this structure improves achievable welfare relative to the adversarial benchmark. Essentially, their solution---and the typical approach to online stochastic matching---can be thought of as solving a utilitarian maximization problem for a large market problem (referred to as a "fluid problem", see \qcite{aouad2022nonparametric} for further discussion). The algorithms use the fluid solution to determine matchings upon agents' arrivals. Competitive equilibria need not be utilitarian efficient except under special preference assumptions, implying these approaches are fundamentally distinct from ours. \qcite{alaei2012online} establishes an algorithm with the intepretation that it utilizes posted prices to allocate arrivals, achieving a perfomance bound of $1 - 1/\sqrt{k+3}$, where $k$ is the minimum number of objects an arrival demands.

\qcite{gao2021} use equilibria to allocate objects in the \textit{online fair division} problem where there is a static population of agents that must be allocated arriving objects, and the market designer would like to achieve efficient and envy-free allocations. Many impossibility results from online matching fade because agents are typically assumed to be non-satiated, ruling out budget or demand constraints. However, the algorithms in this strand use linear programs that explicitly represent competitive equilibria in fluid problems (\qcitenp{benade2024fair}). While the results in online matching and fair division exploit stochastic arrivals, they continue to model preferences as cardinal weights or binary edges and allow the planner to reject arrivals or otherwise feature non-satiated demand. In addition to modeling online matching with ordinal preferences and greedy constraints, we also prove that our mechanism satisfies fairness and strategyproofness properties.

\textbf{Large Markets.} \qcite{azevedo2016supply} presented one of the first comprehensive arguments for the validity of large markets as approximations of finite markets. They show that, in matching with two-sided preferences (e.g. students and schools in school choice), the set of competitive equilibria correspond to the set of stable matchings, and the latter contains a unique matching under strict preferences and assumptions on aggregate demand. Others (\qcite{kojima2009incentives}, \qcite{lee2016incentive}, \qcite{azevedo2019strategy}) focus on strategic incentives in large markets as opposed to our primary criteria of asymptotic efficiency. They prove that, in many settings, large markets deliver (approximate) strategyproofness. Under the assumption of strict preferences, \qcite{liu2016ordinal} proves that all asymptotically efficient, asymptotically strategyproof, and symmetric mechanisms produce equivalent allocations.

Relative to the existing literatures, our approach differs in three key respects. First, we consider a \textit{dynamic} process wherein a greedy allocation constraint requires arrivals to be matched immediately to the best possible object whenever feasible, reflecting institutional features of many applications. Second, we evaluate welfare using stochastic dominance rather than cardinal objectives or worst-case competitive ratios. Third, we integrate pseudomarket techniques into an online setting, yielding a mechanism that is greedy, equal-type envy-free, asymptotically efficient, and strategyproof with arbitrarily high probability. To our knowledge, this is the first online matching mechanism to jointly satisfy these properties under ordinality.

We proceed as follows. In Section 2, we detail our model preliminaries, the large market problem, and the online matching problem. In Section 3, we solve the large market problem. In Section 4, we adapt the large market solution to design an asymptotically optimal mechanism for online matching. In Section 5, we discuss extensions, most prominently to fully online matching where objects also arrive stochastically. In Section 6, we simulate the performance of our mechanism in an applied setting. Finally, we conclude in Section 7.

%% file: Paper/Preliminaries.tex
\section*{2 \hspace{5pt} Preliminaries}
A market is $\mathcal{M} = (X, S, I, T, F)$.
\begin{enumerate}
    \item \textbf{Objects $(X)$:} There is a finite set of objects $X = \{1, 2, ..., |X|\}$ with a typical $x \in X$. The supply of goods is $S = (s_x)_{x \in X}$ where $s_x \in \mathbb{N}_{+}$. There is a null option $o \in X$ that represents receiving no object such that $s_o > |I|$.
    \item \textbf{Agents $(I)$:} There is a finite set of agent \textit{types} $I = \{1, 2, ..., |I|\}$ with a generic $i \in I$. (Mostly, we will write 'agent', but the correct interpretation is an agent type.) Every agent type is associated with ordinal preferences $\succeq_i \hspace{3pt} \in \mathcal{L}(X)$ where $\mathcal{L}(X)$ is the set of all complete, transitive preference relations over $X$. The strict portion of $\succeq_i$ is $\succ_i$ such that $x \succ_i y$ if and only if $x \succeq_i y$ and not $y \succeq_i x$. Each agent type has an arrival time $t_i$, i.e., an agent type is described by preferences and arrival time. The sets $X$ and $I$ are fixed and known.
    \item \textbf{Arrivals $(T, F)$:} The time horizon is a finite set $[T] = \{1, 2, ..., T\}$ with a typical time $t \in [T]$. The market unfolds over this finite horizon, and the arrival process is described by $\mathbf{f} = (f^t(i))_{t \in [T]}$ where $f^t(i)$ is the probability that type $i$ arrives at time $t$. We will write $f^{t,k} = \sum_{n = t}^k f^n(i)$ and $f(i) = \sum_{t \in [T]} f^t(i)$. $\mathbf{f}$ satisfies: if $f^t(i) > 0$, then $t_i = t$. The cumulative distribution functions are $F = (F^t(i))_{t \in [T]}$ which describe the probability that agents up to $i$ arrive at time $t$; we normalize to one arrival per period so that $F^t(|I|) = 1$ for all $t \in [T]$.
\end{enumerate}

We can populate $X$ and $I$ with an arbitrary number of elements and simply consider them to be fixed. In addition, we will hold $T$ fixed. We refer to the market \textit{fundamentals} as $\nu = (F, S)$\footnote{$F$ will be a sufficient statistic for aggregate demand. The fundamentals represent the classical primitives that determine competitive equilibria: demand and supply.}. 

In addition, we define the following notions:

\begin{enumerate}[start=4]
    \item \textbf{Allocations:} An agent's allocation is $a_i \in [0,1]^{|X|}$ satisfying $\sum_{x \in X} a_{i,x} = 1$. Note that such allocations are then probability distributions over objects. A market allocation is a function specifying an allocation for each agent type: $\mathbf{a} : I \rightarrow [0,1]^{|X|}$. $\mathbf{a}$ is feasible if $\sum_{i \in I} \mathbf{a}(i) f(i) \leq S$. 
    \item \textbf{Stochastic Dominance:} An allocation $a'_i$ stochastically dominates $a_i$ if, for each $x \in X$, $\sum_{y \succeq x} a'^{,y}_i \geq \sum_{y \succeq x} a_i^y$. We write that $a'_i \succeq_i a_i$ to represent the (partial) extension of $i$'s preferences over lotteries defined by stochastic dominance. If this is strict for some $x$, then $a'_i \succ_i a_i$, i.e., $a'_i$ strictly stochastically dominates $a_i$.
\end{enumerate}

An allocation $\mathbf{a}$ is \textit{ordinally efficient} if there does not exist any feasible allocation $\mathbf{a}'$ such that, for all $i \in I$, $\mathbf{a}'(i) \succeq_i \mathbf{a}(i)$, and for some $i' \in I$, $\mathbf{a}'(i') \succ_i \mathbf{a}(i')$.

The greedy allocation constraint we consider requires giving an arrival one of her favorite objects that she could possibly receive upon arrival. An allocation $\mathbf{a}$ is \textit{greedy} if, for any $j \in I$, $\mathbf{a}(j)_x > 0$ implies that, for any $x^* \in X$ such that $x^* \succ_j x$, $\sum_{i \in I} \mathbf{a}(i)_{x^*} f^{1,t_j}(i) = s_{x^*}$.

Equivalently, if an agent does not receive a lottery over her most preferred objects, then any object $x^*$ more preferred than an object $x$ she receives with positive probability must have its remaining supply consumed by (weakly) earlier arrivals.

\subsection*{2.1 \hspace{5pt} The Online Matching Problem}

In the online matching problem, one agent $i$ arrives at period $t$ according to $f^t$. We denote the arrival at time $t$ as $\tilde{I}^t$, then, for $k \geq n$, the multi-set of arrivals between $n$ and $k$ is $\tilde{I}^{n,k} = \cup_{n \leq t \leq k} \tilde{I}^t$ and $\tilde{I} = \tilde{I}^{1,T}$. The set of realized allocations at time $t$ is:
\[\tilde{A}^t = \{\tilde{a} : \tilde{a}_i \in [0,1]^{|X|} \text{ for all } i, \tilde{a} \text{ is feasible and } \tilde{a}_{i,x} > 0 \text{ for any } x \in X \implies i \in \tilde{I}^{1,t}\}\]
An \textit{ex-post allocation} is some $\tilde{a} \in \tilde{A}^T$. Note that this can be a lottery. The critical distinction is that an ex-post allocation can only allocate objects to realized agents. Our asymptotic analysis shows that we can implement any lottery using deterministic allocations so that allowing lotteries does not impact our main results. $\tilde{a}$ is ordinally efficient if there does not exist another ex-post allocation $\tilde{a}'$ such that $\tilde{a}'_i \succeq_i \tilde{a}_i$ for all $i \in \tilde{I}$ and $\tilde{a}'_{i'} \succ_{i'} \tilde{a}_{i'}$ for some $i' \in \tilde{I}$. 

An \textit{online matching mechanism} assigns a lottery to realized arrivals given preference reports; it is a function, for each $t \in [T]$, $\pi^t : \mathcal{M} \times \tilde{I}^{1,t} \rightarrow \tilde{A}^t$, and we write $\pi = (\pi^t)_{t \in [T]}$. We allow online matching mechanisms to depend on the market, which, in online allocation, is ex-ante knowledge about objects, preferences, budgets, and the distribution of arrivals. We will usually write $\pi^t_\mathcal{M}(\tilde{I}^{1,t})$ without dependence on $\mathcal{M}$ and $\tilde{I}^{1,t}$, considering them to be implicit. We assume that object allocations are permanent: $\pi^{t+1}_{i,x} \geq \pi^t_{i,x}$ for all $i \in \tilde{I}^{1,t}$ and $x \in X$. $\pi$ is ordinally efficient if there does not exist any ex-post allocation $\tilde{a}$ such that $\tilde{a}_i \succeq_i \pi^T_i$ for all $i \in \tilde{I}$ and $\tilde{a}_{j} \succ_{j} \pi^T_j$ for some $j \in \tilde{I}$. 

An online matching mechanism $\pi$ is greedy if it assigns a most preferred lottery upon an agent's arrival whenever feasible: $\pi^{t_i}_{i,x} > 0$ implies that, for any $x^* \in X$ such that $x^* \succ_i x$, $\sum_{j \in \tilde{I}^{1,t_i}} \pi^{t_i}_{j,x^*} = s_{x^*}$. Last, we say that an online matching mechanism $\pi$ is \textit{strategyproof} if no agent with preferences $\succeq_i$ can receive a strictly better (stochastically dominating) allocation by misreporting her preferences; in this setting, it is equivalent to misrepresenting her identity holding her arrival time fixed: for any $i \in \tilde{I}$, there does not exist any $j \in I$ with $t_j = t_i$ such that $\pi^T_{i}(\{j\} \cup \tilde{I} \setminus \{i\}) \succeq_i \pi^T_{i}(\tilde{I})$\footnote{This does not restrict misrepresentations as we can insert mass zero sets of agents with arbitrary preferences to allow $i$ to report any preferences.}.

\subsection*{2.2 \hspace{5pt} A Basic Impossibility Result}

Generally, it is impossible to design an online matching mechanism that is simultaneously greedy and ordinally efficient. The following example demonstrates this.

\begin{example}\label{example:impossibility}
     There are two objects, $x$ and $y$, each with one unit of supply, and three agent types, $a$, $b$, and $c$. Type $a$ is indifferent between $x$ and $y$, type $b$ prefers $x$ to $y$, and type $c$ prefers $y$ to $x$. One agent arrives per period for two periods. At $t = 1$, agent $a$ arrives with probability one. At period $t = 2$, either agent $b$ or $c$ arrives with probability one-half each. Any greedy allocation must allocate a convex combination of $x$ and $y$ with probability one to agent $a$ at $t = 1$. Therefore, the allocation is ordinally inefficient with positive probability since the arrival at $t = 2$ might not receive her most preferred object with probability one while $a$ is indifferent.
\end{example}

This impossibility motivates our search for an online matching mechanism that will still be desirable given the constraints discussed so far.

%% file: Paper/Results.tex
\section*{3 \hspace{5pt} The Large Market Problem}

Our goal is to design an online matching mechanism---a direct revelation mechanism where agents report preferences upon arrival and receive allocations---that satisfies efficiency desideratum while respecting the greedy allocation constraint. Our approach to accomplish this is to develop a large market equilibrium concept that can be computed at each time period and used to guide allocations. In this section, we define our equilibrium concept, prove existence, and analyze its efficiency and fairness properties.


\subsection*{3.1 \hspace{5pt} Price Equilibria}

We describe our equilibrium concept. A \textit{random price} is a random vector $\mathbf{p} = (\mathbf{p}_x)_{x \in X}$ where $\mathbf{p}_x = p_x + \xi_x$ for some $p_x \in \mathbb{R}_{+}$ and $\xi_x \in \mathbb{R}$. We write $\xi = (\xi_x)_{x \in X}$. $\xi$ is random vector distributed $G(\cdot)$. A budget function $b : I \rightarrow \mathbb{R}_{++}$ (writing $b(i) = b_i$ for short) specifies a strictly positive amount of token money given to each $i \in I$. For a realization $\mathbf{p}$, $i$'s \textit{demand} is her set of cheapest, most preferred objects that are affordable:
\[D_i(\mathbf{p}) = \min_{\mathbf{p} \cdot a_i} \bigg\{ \max_{\succeq_i} \Big\{ a_i \in X : \mathbf{p} \cdot a_i \leq b_i \Big\} \bigg\} \]
We will need to integrate over $\mathbf{p}$ to obtain expected demand. The suitable tool for integration over correspondences is the Aumann integral, defined for any measurable subset $S$:
\[\int_{S} D_i(\mathbf{p}) dG(\xi) = \bigg\{ \int_{S} g_i(\mathbf{p}) dG(\xi) : g_i(\mathbf{p}) \in D_i(\mathbf{p}) \text{ for almost all } \mathbf{p} \bigg\}\]
Aumann integrals are multi-valued functions (correspondences). Henceforth, it is assumed that every integral is Aumann integration, and we suppress dependence on $\xi$ where it is clear. 

$i$'s \textit{lottery demand} is $L_i(p;\xi) = \int_{R_+} D_i(p + \xi) dG(\xi)$. Lottery demand is the expectation over optimal objects consumed at each realization of the correlated prices $\mathbf{p}$. Aggregate demand for times between $n$ and $k \geq n$ is $\mathcal{D}^{n,k}(p) = \sum_{i \in I} L_i(p) f^{n,k}(i)$. Aggregate demand for all times $k \geq t$ is defined as $\mathcal{D}^{t}(p) := \mathcal{D}^{t,T}(p)$. 

An equilibrium will require the market to clear in the sense that aggregate demand does not exceed supply. Our equilibrium will also be specified to consider aggregate demand from certain time periods and onward to allow for implementation in online matching, which we describe in Section 4.

\begin{definition}
    A price $p^*$ is a $(t)$-\textit{price equilibrium} (or simply $(t)$-equilibrium) if:
    \begin{enumerate}
        \item $\mathcal{D}^t(p^*) \leq S$
        \item $p^*_x > 0$ implies $\mathcal{D}^t_x(p^*) = s_x$ for all $x \in X$
    \end{enumerate}
\end{definition}
We say that a $(1)$-equilibrium is simply an equilibrium. We refer to an allocation $\mathbf{a}^*$ as a \textit{$(t)$-equilibrium allocation} if there exists a $(t)$-equilibrium $p^*$ such that, for $i$-ae, $\mathbf{a}^*(i) \in L_i(p^*)$. Similarly, an equilibrium allocation is a $(1)$-equilibrium allocation.

\subsection*{3.2 \hspace{5pt} Existence}

In this sub-section, we establish existence results for our equilibrium concept. In general, equilibria often fail to exist for indivisible object allocation without transfers if one desires exact market clearing and deterministic allocations. The simplest illustration of this is seen in Figure \ref{fig:non-existence}. Aggregate demand is discontinuous at the budget threshold, and any price vector that attempts to clear the market for $x$ either over- or under-allocates it.

\begin{figure}
    \centering
        \begin{tikzpicture}
            \begin{axis}[
                width=10cm,
                height=6cm,
                xlabel={$p_x$},
                ylabel={Probability of $x$},
                xtick={0,0.5,1,1.5,2},
                ytick={0,0.4,0.8,1.2,1.6,2},
                xmin=0, xmax=2.1,
                ymin=0, ymax=2.1,
                grid=major,
                axis lines=left,
                legend pos=north east
            ]

            \addplot[blue, very thick] coordinates {(0,2) (1,2)};
            \addlegendentry{$\int_{I^{1,2}} D_i(p) dF(i)$}

            \addplot[blue, very thick, forget plot] coordinates {(1,0.01) (2,0.01)};

            \addplot[yellow, very thick] coordinates {(0,1) (2,1)};
            \addlegendentry{$S_x$}

            \addplot[blue, dashed] coordinates {(1,0) (1,2)};
            \addlegendentry{$p_y = 1$}
            \end{axis}
        \end{tikzpicture}
    \caption{Non-Existence. $X = \{x,y\}$, $b_i = 1$, and $x \succ_i y$ for all $i \in I$; $T = 2$.}
    \label{fig:non-existence}
\end{figure}
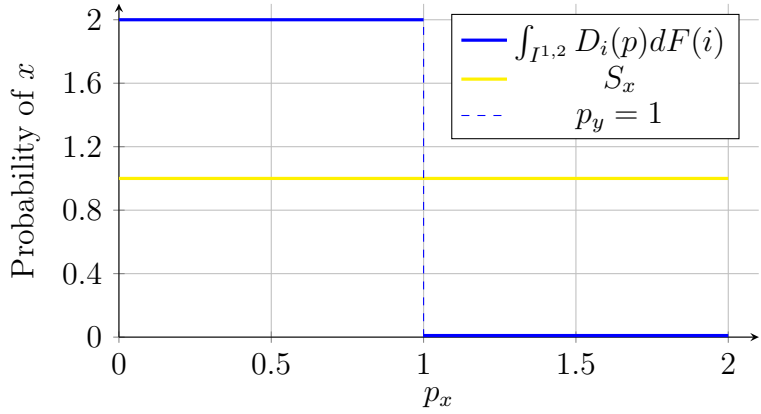

Instead, the market designer would prefer continuous demand to find $p^*$ where $\mathcal{D}^1(p^*) = S$ (Figure \ref{fig:existence}). Price equilibria do not necessarily exist for any random prices, mirroring \qcite{nguyen2025efficiency}. However, as long as the price shock $\xi$ and budget function $b$ satisfy important assumptions, we can guarantee existence.
\begin{assumption}\label{assumption:continuous}
    The price shock $\xi_x$ is continuously distributed for each $x \in X$ with $E[\xi_x] = 0$ and bounded with $\xi_x \in [\underline{\xi}, \bar{\xi}]$.
\end{assumption}
The continuous distribution and boundedness are substantive, but the mean-zero normalization is for convenience. We also assume that budgets are strictly positive and bounded:
\begin{assumption}\label{assumption:budgets}
    $\underline{B} \leq \inf_{i \in I} b_i$ and $\sup_{i \in I} b_i \leq \bar{B}$ for some finite $\underline{B}, \bar{B} \in \mathbb{R}_{++}$.
\end{assumption}
and
\begin{assumption}\label{assumption:bounded}
    $\xi$ is \textit{bounded away from budgets}: $\bar{\xi} < \underline{B}$.
\end{assumption}
Now, we can state our result.

\begin{theorem}\label{theorem:existence}
    Under Assumptions \ref{assumption:continuous} and \ref{assumption:budgets}, a $(t)$-price equilibrium exists for every $t \in [T]$. In addition: under Assumption \ref{assumption:bounded}, there exists some budgets $b$ such that every price equilibrium allocation is greedy.
\end{theorem}

Demand is upper hemicontinuous at all prices that do not exactly lie on the boundary of the budget set which is a measure-zero set. It follows that the Aumann integral is upper hemicontinuous. In addition, Aumann integration also induces convex-valued lottery demand. This allows us to apply Kakutani's Fixed Point Theorem to prove existence. We note that the first part of Theorem \ref{theorem:existence} also applies for combinatorial demand. Our work does not focus on this, but it might be of independent interest for combinatorial allocation with indifferences.

\begin{figure}[t]
    \centering
    \begin{tikzpicture}
            \begin{axis}[
                width=10cm,
                height=6cm,
                xlabel={$p_x$},
                ylabel={Probability of $x$},
                xtick={0,0.5,1,1.5,2},
                ytick={0,0.4,0.8,1.2,1.6,2},
                xmin=0, xmax=2.1,
                ymin=0, ymax=2.1,
                grid=major,
                axis lines=left,
                legend pos=north east
            ]

            \addplot[red, very thick, smooth] coordinates {(0,2) (0.2,1.8) (0.4,1.6) (0.6,1.4) (0.8,1.2) (1,1)};
            \addlegendentry{$\mathcal{D}^1(p)_x$}
                
            \addplot[red, very thick, forget plot] coordinates {(1,0) (2,0)};
                
            \addplot[red, very thick, forget plot] coordinates {(1,0) (1,1)};

            \addplot[black, forget plot, mark=*, mark size=3pt, only marks] coordinates {(1,1)};

            \addplot[yellow, very thick] coordinates {(0,1) (2,1)};
            \addlegendentry{$S_x$}

            \addplot[blue, dashed] coordinates {(1,0) (1,2)};
            \addlegendentry{$p_y = 1$}
            \end{axis}
        \end{tikzpicture}
    \caption{Existence. $X = \{x,y\}$, $b_i = 1$, and $x \succ_i y$ for all $i \in I$; $T = 2$.}
    \label{fig:existence}
\end{figure}
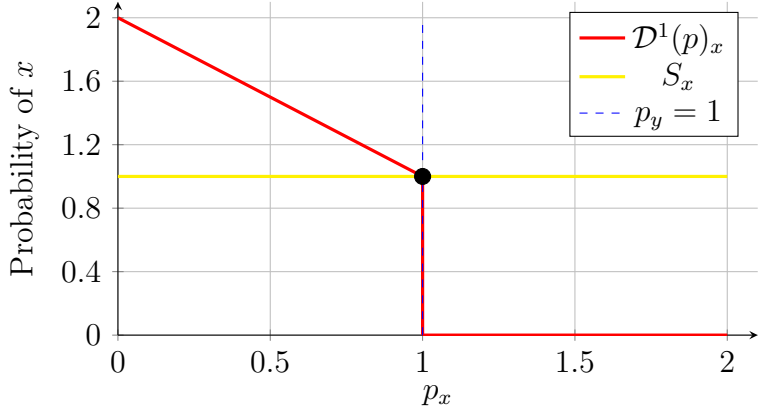

This implies that the market designer can introduce an arbitrarily small amount of noise into the price vector and solve for price equilibria. The Birkhoff-von Neumann Theorem guarantees that market clearing lotteries can be implemented as randomizations over ex-post feasible and efficient allocations. While this is true in the large market, it is not in the online problem, and we demonstrate how to build on this approach in Section 4 for the latter.

The greedy property follows from constructing budgets that decrease over time. If income is sufficiently asymmetric, then whenever there is a realization $\mathbf{p}$ such that an agent $i$ cannot purchase an object $x^*$ that $i$ prefers to an object $x$, it follows that any agent $j$ with $t_j > t_i$ will never be able to purchase $x^*$ since the price shock is bounded. The proof provides the exact expression to construct greedy equilibria: defining $b_i := \tilde{b}^{t_i}$ where $\tilde{b} = (b^1, b^2, ..., b^T)$, then $p^*$ induces a greedy allocation if $\tilde{b}^n > \tilde{b}^{n+1} + \bar{\xi} + \underline{\xi}$.

\subsection*{3.3 \hspace{5pt} Efficiency and Fairness}

Next, we would like to understand when price equilibrium allocations are ordinally efficient. Assumptions \ref{assumption:continuous} through \ref{assumption:bounded} are not sufficient. We state two independent assumptions:
\begin{assumption}\label{assumption:perfect-correlation}
    \textit{No Tie-Breaking (NTB):} $\xi_x = c$ for some $c \in \mathbb{R}$ such that $c \sim H(\cdot)$.
\end{assumption}
and
\begin{assumption}\label{assumption:strong-correlation}
    \textit{Random Tie-Breaking (RTB):} $\xi_x = c + \zeta_x$ for some $c \in \mathbb{R}$ such that $c \sim H(\cdot)$ and $\zeta_x \sim H_x(\cdot)$ independent across each $x \in X$ such that $\zeta_x \in [-z, z]$.
\end{assumption}
It is easier to understand Assumption \ref{assumption:perfect-correlation} in light of Assumption \ref{assumption:strong-correlation}. The latter is a \textit{random tie-breaking} assumption. Continuous, random, and independent price variation across objects causes the event of exact price ties to occur with probability zero. Agents demand the cheapest, most preferred objects. Under Assumption \ref{assumption:strong-correlation}, this implies that agent demand is always single-valued even though, in principle, an agent could be indifferent between objects. It is important to note that this does not imply that we find equilibria for strict preferences with randomly broken ties. Instead, tie-breaking is only relevant when prices for two goods are sufficiently close. 

In contrast, the former assumption is \textit{no tie-breaking}. $p_x = p_y$ necessarily implies $\mathbf{p}_x = \mathbf{p}_y$ for any realizations. Theorem \ref{theorem:existence} proves that, in either case, price equilibria exist. However, these two assumptions induce different efficiency results.

We re-state that an allocation $\mathbf{a}$ is \textit{ordinally efficient} if there does not exist any feasible allocation that stochastically dominates $\mathbf{a}$ for all agents and strictly so for some positive measure set of agents. Our next result states the sufficient condition for efficiency.

\begin{theorem}\label{theorem:fwt}
    Under Assumptions \ref{assumption:bounded} and \ref{assumption:perfect-correlation}, every price equilibrium allocation $\mathbf{a}^*$ is ordinally efficient.
\end{theorem}

Some distance away from budgets in $\xi$'s bound is necessary to guarantee that price equilibria are non-wasteful. Otherwise, when $p_x = 0$ for some $x$, it might be possible that the price realization causes $\mathbf{p}_x > b_i$ for some $i$ so that $i$ cannot purchase $x$ even though it is free. Assumptions \ref{assumption:bounded} and \ref{assumption:perfect-correlation} allow us to prove a key lemma that $p^* \cdot \mathbf{a}^*(i) \leq p^* \cdot a_i$ for any $a_i$ that stochastically dominates $\mathbf{a}^*(i)$ (and this relation also holds strictly). The latter assumption is necessary to guarantee that there are no "price reversals", that is, $\mathbf{p}_y > \mathbf{p}_x$ even when $p_x > p_y$. Integrating over all agents then implies that no feasible allocation can stochastically dominate the equilibrium allocation. Assumption \ref{assumption:strong-correlation} also yields efficiency under a specific condition:
\begin{theorem}\label{theorem:ties}
    Under Assumptions \ref{assumption:bounded} and \ref{assumption:strong-correlation}:
    \begin{enumerate}
        \item[(i.)] A price equilibrium allocation $\mathbf{a}^*$ for $p^*$ is ordinally efficient if $p^*_x,p^*_y > 0$ implies that $|p^*_x - p^*_y| > 2z$ for all $x,y \in X$.
        \item[(ii.)] $\mathcal{D}^t(p)$ is a continuous function in $p$.
    \end{enumerate}
\end{theorem}
The second statement in Theorem \ref{theorem:ties} implies that aggregate demand (in fact, all individuals' demands) collapses to a single value under random tie-breaking. The advantage of this tie-breaking method is that equilibrium computation is significantly simplified since it becomes unnecessary to approximate multi-valued demand correspondences. Condition (i) in the Theorem provides a sufficient condition to check that the computation produces an ordinally efficient allocation.

The condition implies that whenever prices are sufficiently close such that aggregate demand considers these objects to be "indifferent", random tie-breaking can cause price reversals that over- or under-price objects. Though this is endogenous, there is good reason to believe that it holds in some markets. We can make $z$ arbitrarily small so that, relative to the common shock $c$, each shock $\zeta_x$ barely impacts aggregate demand for each $x \in X$ with the qualification that prices are not exactly equal. A move from no tie-breaking to random tie-breaking with arbitrarily small $z$ should then be almost exactly market clearing. Further, under the assumption that prices are not exactly equal, aggregate demand (under Assumption \ref{assumption:perfect-correlation}) is locally continuous in prices. Therefore, we conjecture that a small movement in the price should suffice to clear the market exactly. This argument implies that it might be possible to strengthen Theorem \ref{theorem:ties} to $|p^*_x - p^*_y| > 0$.

In contrast, if $\xi_x = \zeta_x$, this intuition is definitively false. The equilibrium price for $x$ is highly sensitive to the distribution of $\zeta_x$ because any positive price for $x$ will concentrate around the budget constraint as variance decreases. For example, suppose that $b_i = 1$ for all $i \in I$. There is no equilibrium where $p_x - z > 1$, and any equilibrium where $1 > p_x + z$ is also an equilibrium for $p_x = 0$. Hence, we can assume that $p_x \in [1 - z, 1 + z]$. We should expect it to be very likely that $|p^*_x - p^*_y| < z$ for two positively priced goods $x,y$. We conclude that it is crucial to appropriately calibrate the price error terms. The usefulness of Assumption \ref{assumption:strong-correlation} will come into focus in the next section as we consider large market asymptotics.

Next, we turn to fairness properties.
\begin{definition}
    An allocation $\mathbf{a}$ is \textit{envy-free} if there does not exist any $i,j \in I$ such that $\mathbf{a}(j) \succ_i \mathbf{a}(i)$.
\end{definition}
It is immediate that greedy allocations generically cannot be envy-free since earlier arrivals will have strict priority over later arrivals. Instead, we consider a relaxation called \textit{equal-type envy-freeness} (ET-EF).
\begin{definition}
    An allocation $\mathbf{a}$ is equal-type envy-free if there exists some complete, transitive order $\eqtrir$ over $I$ with strict relation $\trir$ such that $\mathbf{a}(j) \succ_i \mathbf{a}(i)$ implies that $j \trir i$
\end{definition}
In other words, no agent envies another agent of the same "type" according to $\eqtrir$. Our next result shows that price equilibrium allocations are ET-EF.

\begin{theorem}\label{theorem:fairness}
    Under Assumption \ref{assumption:bounded}, there exists budgets $b$ such that every price equilibrium allocation $\mathbf{a}^*$ is equal-type envy-free.
\end{theorem}

We let the \textit{type} $i$ of an agent be her arrival time $t_i$. Using the above budgets, if $\mathbf{a}^*(j) \succ_i \mathbf{a}^*(i)$, then $p^* \cdot \mathbf{a}^*(j) > p^* \cdot \mathbf{a}^*(i)$, implies that $b_j > b_i$ thus $j \trir i$. Hence, no agent envies another agent of a (weakly) lower type. Note that while this implies that equal-type agents do not envy each other's lotteries, it does not imply that there will be no ex-post envy after resolving the lotteries. Ultimately, this is a fundamental limitation of indivisible object allocation.

\section*{4 \hspace{5pt} The Online Matching Problem}

We turn to our analysis of the online matching problem in this section. Using our large market equilibrium concept, we will construct an online matching mechanism that is ex-post efficient as the number of arrivals per period grows large and is always ex-post greedy.

\subsection*{4.1 \hspace{5pt} The Sequential Equilibrium Mechanism}

The \textit{online empirical distribution} at $k$ is $\tilde{F}(i;k) = (\tilde{F}^1, \tilde{F}^2, ..., \tilde{F}^k, F^{k+1}, ..., F^T)$, which replaces the conditional distribution at times $t \leq k$ to match the realized arrivals: $\tilde{F}^k(i) = \sum_{j \in \tilde{I}^k} \mathbbm{1} \{j \leq i\}/|\tilde{I}^k|$.

Our method proposes an online matching mechanism that leverages price equilibria to achieve asymptotic efficiency while respecting the greedy constraint ex-post. In particular, we compute a large market equilibrium $p^t$ at each time $t$ then randomly allocate objects to arrivals according to their lotteries. The formal description of the Sequential Equilibrium Mechanism (SEM) is: \\~\

\textit{Algorithm: Sequential Equilibrium Mechanism}
\begin{enumerate}
    \item If $t = 1$, initialize $S^1 := S$. At time $t$, compute a price equilibrum $p^t$ for the fundamentals $\nu^t = (\tilde{F}(\cdot;t), S^t)$.
    \item Draw an ex-post allocation $\tilde{a}^t$ from an equilibrium allocation $\mathbf{a}^t$ and allocate $\tilde{a}^t(\tilde{i})$ to each $\tilde{i} \in \tilde{I}^t$. Set $S^{t + 1} := S^t - \sum_{i \in \tilde{I}^t} \tilde{a}^t(i) \tilde{f}^t(i)$.
    \item Continue to $t + 1$, or terminate if $t = T$ or $S^{t + 1} = \{0\}^{|X|}$. \\~\
\end{enumerate}

SEM defines a class of mechanisms that depend on the equilibrium selection at each period $t$. Our main result is that SEM achieves all of our desiderata if the underlying market satisfies a key condition.

We say that a mechanism $\pi$ is \textit{asymptotically efficient} if its ex-post allocation converges to an ordinally efficient allocation as the market size grows large. Specifically, \textit{replica markets} are a sequence of markets $(\mathcal{M}_n)_{n \in \mathbb{N}}$ such that each market has the same fundamentals but with $n$ times the mass of agents and supply: $\mathcal{M}_n = (X, nS, I, T, nF)$. The realized arrivals for replica $n$ at time $t$ are $\tilde{I}^t(n)$ (a multi-set). Define $\pi_n = (\pi_n^t)_{t \in [T]}$ by $\pi_n^t := \pi_{\mathcal{M}_n}^t(\tilde{I}^{1,t}(n))$.
\begin{definition}
    An online matching mechanism $\pi$ is asymptotically efficient if, for all $\epsilon > 0$:
    \[\lim_{n \rightarrow\infty} \Pr\bigg(\|\pi_n - \tilde{a}^*\|_\infty < \epsilon \bigg) = 1\]
    for some ordinally efficient allocation $\tilde{a}^*$\footnote{The infinity norm $\|\cdot\|_\infty$ is defined as the maximum absolute value of the components of a vector. See the Appendix for further details.}.
\end{definition}

The $(t)$-price equilibria for fundamentals $\nu$ are $\mathcal{P}^t(\nu)$. Similarly, the set of equilibria allocations given fundamentals $\nu$ are $\mathcal{E}^t(\nu)$.
\begin{definition}
    The fundamentals $\nu$ are regular if, for all $t \in [T]$, there exists some open ball $U$ in the metric space defined by $\|\cdot\|_\infty$ around $\nu$ and a finite number $N$ of continuous functions $\lambda_1(\mu), \lambda_2(\mu), ..., \lambda_N(\mu)$ from fundamentals to equilibrium prices such that for any $\mu \in U$:
    \begin{enumerate}
        \item[(i.)] $\lambda_n(\mu)$ is non-empty for each $n$.
        \item[(ii.)] If $p \in \mathcal{P}^t(\mu)$ then $p \in \lambda_n(\mu)$ for some $n$.
    \end{enumerate}
    A market is regular if its fundamentals are regular. In contrast, a market is \textit{allocation regular} (for short, $\ell$-regular) we can replace the above with maps from fundamentals to allocations such that (ii): if $\mathbf{a} \in \mathcal{E}^t(\mu)$ then $\mathbf{a} \in \lambda_n(\mu)$ for some $n$.
\end{definition}
Regularity is satisfied for almost all fundamentals (\qcitenp{debreu-1970}). For the following result, we let Assumptions \ref{assumption:continuous} through \ref{assumption:bounded} hold.

\begin{theorem}\label{theorem:sem}
    SEM is equal-type envy-free, strategyproof, and greedy. Moreover:
    \begin{enumerate}
        \item[(i.)] Under Assumption \ref{assumption:perfect-correlation}, if $\mathcal{M}$ is $\ell$-regular, then there exists a selection of SEM that is asymptotically efficient.
        \item[(ii.)] Under Assumption \ref{assumption:strong-correlation}, if $\mathcal{M}$ is regular and the condition in Theorem \ref{theorem:ties} holds, then there exists a selection of SEM that is asymptotically efficient.
    \end{enumerate}
\end{theorem}

$\ell$-regularity is necessary for stability of the large market convergence process. Under random tie-breaking, $\ell$-regularity collapses to regularity\footnote{The proof is immediate: every agent's demand is continuous in equilibrium prices by Theorem \ref{theorem:ties}, therefore continuity of the equilibrium prices commutes to continuity of equilibrium allocations.}. Briefly, we note that $\ell$-regularity is likely to hold generically under no tie-breaking. Existing results have shown that essentially every economy is regular (\qcitenp{debreu-1970}). $\ell$-regularity follows if either (i) there are no price ties in equilibrium or (ii) if there are price ties, local perturbations in fundamentals preserve the ties. In either case, for sufficiently close prices, all agents demand similar allocations. Hence, regularity commutes to $\ell$-regularity. We conduct simulations in the Appendix that show that 99.9\% of perturbations satisfy condition (ii) in randomly generated markets.

We propose that asymptotic efficiency is a solid foundation for utilizing SEM in online allocation. Example \ref{example:impossibility} demonstrates that no mechanism can generically be ex-post ordinally efficient. Asymptotic efficiency implies that SEM is "correct" in the sense that, absent the feature of the environment that renders ex-post ordinal efficiency impossible (finiteness), SEM is ex-post ordinally efficient. Moreover, its performance is likely to improve in larger markets. Nevertheless, we acknowledge that this is not an airtight rationale; for small $n$, it is possible that SEM will perform poorly. Our primary focus in our simulations is to analyze SEM's empirical effects versus a status-quo mechanism.


The critical facet of the design of SEM that guarantees greediness and asymptotic efficiency simultaneously is its reliance on sequential equilibrium. It is not feasible to compute a $p^1$ equilibrium to guide the online allocation because when the realized arrivals deviate from the large market distribution, some objects could become over- or under-demanded. The period-$1$ equilibrium allocation does not internalize these disturbances, and, as a result, the ex-post allocations can fail to be greedy. Sequential equilibria resolve this issue through updating the supply vectors after each period which then necessitates re-computing price equilibria. Together with stability of the equilibria allocations, we can achieve convergence to a $(1)$-equilibrium allocation, which is ordinally efficient under either set of assumptions.

Moreover, any deterministic allocation in the support of an ordinally efficient allocation is Pareto efficient (\qcitenp{ramezanian2022robust}), and ordinal efficiency is invariant to the support of objects allocated to agents (\qcitenp{alva2024efficiency}). This suggests that the probability that any SEM deterministic allocation is inefficient also converges to zero.

Of the properties in Theorem \ref{theorem:sem}, strategyproofness only holds under the assumption that $|\tilde{I}^t| = 1$ for all $t$. Nevertheless, SEM is also \textit{strategyproof with probability one} under any sets of realized arrivals as $n$ grows large. Any allocation that an agent can receive by misreporting either (a) converges to an allocation that does not stochastically dominate the allocation she receives from truthfulness or (b) converges to the agent's allocation under truthfulness.

$i$'s allocation under $\pi_n$ when truthful is $\pi_{n,i} := \pi_{n,i}^T(\tilde{I}(n))$. $i$'s allocation when misrepresenting as an agent $j$ is $\pi_{n,(i,j)} := \pi_{n,i}^T(\{j\} \cup \tilde{I}(n) \setminus \{i\})$.

\begin{definition}
    $\pi$ is \textit{strategyproof with probability one} (SP1) if, for any $\epsilon > 0$:
    \[\lim_{n \rightarrow \infty} \Pr\bigg(\pi_{n,i} \succeq_i \pi_{n,(i,j)} \lor \|\pi_{n,i} - \pi_{n,(i,j)}\|_\infty < \epsilon \bigg) = 1\]
    for all $i \in \tilde{I}(n)$ for any misrepresentation $j \in I$ such that $t_j = t_i$.
\end{definition}

\begin{corollary}
    Under the conditions for Theorem \ref{theorem:sem}, there exists a selection of SEM that is strategyproof with probability one.
\end{corollary}

Theorem \ref{theorem:sem} implies this result. An agent has zero price impact in the large market. Equivalently, she faces fixed prices. Misrepresenting herself as another agent that has the same budget (equal arrival time) cannot give her a strictly stochastically dominating allocation. Otherwise, she would purchase that agent's allocation. The replica market allocations converge to these allocations given $\ell$-regularity, and the Corollary follows.

\subsection*{4.2 \hspace{5pt} Comparison to Existing Mechanisms}

We discuss the relationship between SEM mechanisms and others in the literature. As far as we are aware, SEM is the only mechanism for online stochastic matching that satisfies ordinality and the greedy constraint. Our comparisons will therefore focus on (a) the solution to the static problem subject to the greedy and ordinality constraints (Serial Dictatorship with Indifferences), (b) solutions to the static problem with the ordinality constraint and without the greedy constraint (CERI, ACEEI, and EPS), and (c) briefly discuss solutions to the online problem subject to the ordinality constraint.

\textbf{Serial Dictatorship with Indifferences (SDI).} SDI mechanisms are extensions of Serial Dictatorships (SD) to allocate objects in static problems when agents have indifferences\footnote{See \qcite{highsmith2024matching} for further discussion.}. They are deterministic and Pareto efficient but not (equal-type) envy-free. Generally, randomized SD mechanisms are not ordinally efficient. Seeing as SDI mechanisms are equivalent to SD mechanisms on strict preference domains, it follows that randomized SDI mechanisms are not ordinally efficient.

Directly using SDI to allocate objects is infeasible; the matchmaker would need to wait until all arrivals are realized because SDI determines the match for a dictator using the preferences of later dictators. This information is necessary to efficiently allocate objects within an agent's indifference set. However, this violates the greedy allocation constraint. If the market designer had omniscient foreknowledge of the agent arrivals $\tilde{I}$, the designer could obtain a Pareto efficient deterministic allocation by running SDI then allocating to each agent the object that she would get under SDI when the agent arrives. The allocation would also be greedy (because earlier arrivals are earlier dictators).

Interestingly, in the case of one arrival per period, one can think of SEM as a SDI mechanism. The prices of SEM serve as a proxy for how SDI would break ties within indifference sets absent ominiscience. An arriving agent receives her favorite object with the lowest price. They are not synonymous with more than one arrival per period: SEM will allocate envy-free lotteries whereas SDI does not.

\textbf{Approximate Competitive Equilibrium from Equal Incomes (ACEEI)} and \textbf{Competitive Equilibrium from Random Incomes (CERI).} Both ACEEI (\qcitenp{budish2011combinatorial}) and CERI (\qcite{nguyen2025efficiency}) are mechanisms primarily intended for combinatorial allocation. ACEEI is approximately Pareto efficient (ex-post) and envy-free up to one good (ex-post). In comparison, CERI is ordinally efficient and envy-free (per our definition of envy-freeness), whereas ACEEI is not. Our solution to the large market problem renders SEM more akin to CERI for a static problem, in fact, we conjecture that the two are equivalent on the strict preference domain despite our usage of random prices versus CERI's random budgets. Yet, ACEEI and CERI differ from SEM in the crucial respect that they are defined only for this strict preference domain whereas we allow for the general preference domain with indifferences.

\textbf{Extended Probabilistic Serial (EPS).} \qcite{katta2006solution} develop this method to apply \qcite{bogomolnaia2001probabilistic}'s Probabilistic Serial (PS) to the general preference domain. An EPS mechanism solves a network flow problem that is similar to the original PS formulation of "eating speeds." They prove a First and Second Welfare Theorem for EPS mechanisms; that is, EPS mechanisms characterize the set of ordinally efficient allocations. If our conjecture that price equilibria are analogous to CERI in the general preference domain is correct, then it is likely that price equilibria also characterize the set of ordinally efficient allocations just as CERI do under the strict domain. Therefore, price equilibria and EPS mechanisms would be equivalent under unit demand. SEM is distinct for two reasons:

(i) Our flexibility in specifying the price error term implies that price equilibria can implement different tie-breaking behavior which has interesting implications for ordinal efficiency (Theorem \ref{theorem:ties}) and asymptotic convergence (Theorem \ref{theorem:sem}).

(ii) EPS is for static problems, and it is not clear if EPS \textit{can} be extended to online matching. A "dynamic" EPS (DEPS) must compute allocations not just for the static market but also for expected arrivals so that earlier arrivals receive their preferred objects over later arrivals. Additionally, DEPS would require a feasible iterative implementation as in Lemma \ref{lemma:continuation}. Last, extending this hypothesized DEPS to fully online matching as in Theorem \ref{theorem:fully-online} could be significantly more difficult if not impossible. We view proving or contradicting these equivalence results as interesting open problems.

\textbf{Online Matching with Ordinality.} \qcite{hoefer2017combinatorial} represents the modelling of online matching with ordinality. They assume underlying cardinal valuations, but the market designer only has access to an ordinal ranking consistent with the cardinal values. The goal is to identify algorithms that "lose little" compared to knowing the cardinal values. Our approach is fundamentally different as our welfare criteria is stochastic dominance, not performance benchmarks against cardinal loss.

\section*{5 \hspace{5pt} Extensions}

In this section, we explore an extension to our model to allow for stochastic arrivals of objects.

\subsection*{5.1 \hspace{5pt} Fully Online Matching}

A \textit{fully online matching problem} refers to stochastic arrival of supply and demand. We focus our analysis on the description of this problem for the large market. The implementation for the online problem is likely to follow the same steps as our main model, though we later note a few subtleties that must be handled to extend our results.

A large market in the fully online matching problem is $\mathcal{M} = (X, S, I, T, F)$:
\begin{enumerate}
    \item \textbf{Objects} ($X$), \textbf{Agents} ($I$), and \textbf{Arrivals} ($T,F$) remain as before.
    \item \textbf{Supply} $S = (S^t)_{t \in [T]}$ is a sequence of random variables $S^t \in \mathbb{R}^{|X|}$ with joint CDFs $(H^t(\cdot))_{t \in [T]}$ such that $S^t \sim H^t$.
\end{enumerate}

A \textit{time tax} price vector is $\mathbf{p} = (\mathbf{p}^t)_{t \in [T]}$ where $\mathbf{p}^t \in \mathbb{R}^{|X|}$ and $\mathbf{p}^t = p^t + \xi$. We maintain Assumptions \ref{assumption:continuous} through \ref{assumption:bounded}. An agent $i$ with arrival time $t_i$ faces time-specific prices:
\[D_i(\mathbf{p}) = \min_{\mathbf{p} \cdot a_i} \bigg\{ \max_{\succeq_i} \Big\{ a_i \in X : \mathbf{p}^{t_i} \cdot a_i \leq 1 \Big\} \bigg\} \]
All else remains as in the general model. We denote period $t$ through $k$ expected supply as:
\[\mathcal{S}^{t,k} = \sum_{n = t}^k E[S^n]\]
We can now state our equilibrium concept and result.
\begin{definition}
    A time tax price vector $p^* = (p^{*,t})$ is a $(t)$-Lindahl equilibrium if:
    \begin{enumerate}
        \item $\mathcal{D}^{t,k}(p^*) \leq \mathcal{S}^{t,k}$ for all $k \in [T]$ such that $k \geq t$.
        \item $p^{*,k}_x > 0$ implies $\mathcal{D}^{t,k}_x(p^*) = \mathcal{S}^{t,k}_x$ for all $x \in X$ and $k \in [T]$ such that $k \geq t$.
    \end{enumerate}
\end{definition}
As before, we will refer to a $(1)$-Lindahl equilibrium as simply a Lindahl equilibrium.

\begin{theorem}\label{theorem:fully-online}
    Under Assumptions \ref{assumption:continuous} and \ref{assumption:budgets}, a $(t)$-Lindahl equilibrium exists for every $t \in [T]$. In addition: under Assumption \ref{assumption:bounded}, every Lindahl equilibrium allocation is greedy.
\end{theorem}

This Theorem applies in the degenerate case where $S$ is not stochastic and identical across $t$, implying that Lindahl equilibria generalize price equilibria. The necessity to allow for a "time tax" arises because a single, fixed price vector cannot always satisfy multiple market clearing constraints. A simple case when this ensues is a situation with identical preferences and arrival of supply across time periods, but arrival of demand doubles in period $2$ from period $1$. Any fixed price vector that satisifes market clearing in period $1$ must necessarily induce twice the demand of period $1$ in period $2$, meaning the expected supply will be exceeded by a factor of one-half.

Theorem \ref{theorem:fully-online} establishes the possibility of using a large market approach for fully online matching. Interestingly, there is also little worry about using "personalized" prices in this setting, as an agent's arrival time is verifiable so that a manipulation cannot award her more favorable prices. We add a caveat: the application must fit \textit{immediate departure}. A $(t)$-Lindahl equilibrium ignores agent demand from periods $n < t$. This formulation implicitly assumes that agent departure is immediate. That is, arrivals cannot wait for a future time period to match. Immediate departure is the textbook case in the literature on fully online matching, but its appropriateness varies across applications. 

\section*{6 \hspace{5pt} Simulations}

In this section, we describe our empirical application and present simulation results.

\subsection*{6.1 \hspace{5pt} Institutional Context}

In the United States, the foster care system is responsible for providing care to children found victim of abuse or neglect at the hands of their legal guardians. Each state administers its own foster care program, and, in many cases, local counties govern quasi-independent programs. Upon discovery of substantiated abuse or neglect, the local authority (in coordination with courts) removes the child from the custody of their legal guardians. Subsequently, the local authority assigns a caseworker to represent the child's best interests, and the caseworker identifies a temporary \textit{foster home} to care for the chil. Nationwide, children that enter foster care are at higher risk of negative emotional, behavioral, and biological outcomes (\qcitenp{leve2012practitioner}). In 2018, the Family First Prevention Act was established with one goal of reducing the number of children entering foster care by providing services to at-risk families (\qcitenp{font2020foster}). 

Our partner, the anonymous Firm, is a U.S-based nonprofit whose mission is to deter children from entering foster care. The Firm is organized into local chapters that operate on a state-by-state or county-by-county basis. In the locale specific to our work, The Firm provides temporary hosting services in collaboration with the local authorities. Parents that are in need of social assistance---for example, a mom that needs a caregiver for her children while she attempts to find employment---can contact the Firm to request to have their children hosted in a \textit{host home} that volunteers for the Firm. In addition, the Firm receives referrals directly from the local authorities that involve abuse or neglect cases that are deemed to have low risk of immediate harm to the child. The Firm subsequently identifies one of its own host homes to care for the child and coordinates with the local authorities to manage the child's social services. The Firm's caseworkers vet and approve host homes, process intake of children, and use a centralized system to attempt to match children to host homes. 

The Firm aims to minimize the time that children spend in hostings and provide high-quality assistance to parents in order to achieve quick reunification. The Firm and local authority are aligned in a desire to safely reunify children with biological parents, and the Firm provides economic value in both alleviating need for a state-approved foster home as well as potentially deterring entrance into foster care---which is costly (\qcitenp{barth2006comparison})---altogether. 

We study the matching mechanism used by the Firm to allocate host homes to children. Their context fits our key constraints. The agents, children, arrive stochastically in patterns that often depend on time-based trends (school year, holidays, etc). The host homes are known apriori within planning horizons that span weeks to months. The Firm must match children to host homes immediately upon intake due to legal and logistical constraints. Finally, the Firm implicitly imposes ordinal preferences onto children.

In particular, the Firm is indifferent to placing children in any host home subject to that home's ability to provide adequate care for the child. Moreover, they elicit dichotomous preferences from host homes. This is equivalent to redefining children's preferences so that a child either prefers a home to the outside object or does not, the child is indifferent between any homes preferred to the outside option, and a home is acceptable if and only if it is acceptable to the child and the child is acceptable to the home. Consequently, one can treat homes as objects without preferences\footnote{It is immediate that a matching is ordinally efficient and individually rational if and only if it is ordinally efficient with respect to the problem redefined so that only children have preferences.}. Informal conversations with volunteers in the Firm indicate that home preferences vary substantially over the age of children. A home finds any child below a certain age to be acceptable, and younger children tend to be acceptable more often than older children. This motivates some choices we describe in our simulations. 

The Firm's current mechanism is an informal process where, upon intake of a child, every host home receives a notification describing the child's characteristics and the reason for intake. Host homes then respond if they are willing to host the child, and the caseworker assigns the child to the first willing home. Alternatively, if the caseworker anticipates a need for a specific match fit, the worker may circumvent this process and directly contact a home to request hosting. Even assuming that host homes are perfectly altruistic (desiring to maximize the number of matches) and perfectly informed on future arrivals, this mechanism can fail to produce ordinally efficient allocations due to strategic coordination failures, motivating a case for applying SEM to improve the matching process.

\subsection*{6.2 \hspace{5pt} Market Primitives}

Comparing SEM to the current process requires observational data on realized matches and other model primitives. Currently, we do not have access to the requisite data. In lieu of ongoing work with the Firm, we use a simplified procedure to model their existing matching process and compare it to SEM on a basic two-by-two example. 

We represent the current process using \textit{Serial Dictatorship with Random Tie-Breaking (SD-RTB)}. In SD-RTB, we randomly permute the order of children upon their arrival and sequentially, randomly allocate one of their most preferred homes to each child. This captures the essence of the current process: host homes respond in an arbitrary order, and the caseworker allocates to the first willing home. Clearly, this does not capture other important aspects of the current process, such as caseworker discretion and strategic behavior by host homes. Nevertheless, this simplified model provides a useful benchmark to compare against SEM.

We simulate a market with two host homes ($a,b$) and two types of children ($c_1, c_2$) over $T = 4$ periods. Type 1 is \textit{selective}: $a \succ_{1} o \succ_1 b$. That is, $b$ is unacceptable. This may be interpreted as an older child that fewer homes find acceptable. In contrast, type 2 is non-selective: $\{a,b\} \succ_2 o$, that is, either $a$ or $b$ is acceptable. Each period, $n$ children arrives according to the distribution: $F(c_1) = 0.5$ and $F(c_2) = 0.5$, where we vary $n = \{1,5,10,25\}$ in our simulations. We simulate twenty-five markets for each $n$. Each home has unit supply. Budgets and price error terms are set to satisfy Assumptions \ref{assumption:continuous} through \ref{assumption:perfect-correlation}. We use a customized algorithm tuned for computational speed to implement SEM.

\begin{table}[t]
    \centering
    \caption{Comparing SEM to SD-RTB}
    \begin{tabular} {l@{\hskip 7ex}c@{\hskip 7ex}c@{\hskip 7ex}c@{\hskip 7ex}c}
    \hline
    \hline
    \textit{Mechanism} & $n = 1$ & $n = 5$ & $n = 10$ & $n = 25$ \\ \hline
    \textit{Placement Rates} & & & & \\
    SEM & 0.870 & 0.948 & 0.949 & 0.924 \\
    & (0.179) & (0.060) & (0.064) & (0.084) \\
    & & & & \\
    SD-RTB & 0.760 & 0.860 & 0.852 & 0.820 \\
    & (0.169) & (0.074) & (0.063) & (0.029) \\ \hline
    \end{tabular}
    
    \label{table:sem-sd}
\end{table}

\subsection*{6.3 \hspace{5pt} Results}

SEM is asymptotically efficient; its lotteries will be ordinally efficient as $n$ grows to infinity. However, for small $n$, SEM may perform poorly. We analyze the performance of SEM versus SD-RTB in finite markets of various sizes to evaluate this. Given dichotomous preferences, our main metric is the placement rate, that is, the percentage of children who are matched to a home. Our preliminary results are encouraging in this basic environment.

On average, SEM matches a higher percentage of children than SD-RTB across all market sizes (Table \ref{table:sem-sd}). The performance gap is most pronounced when $n = 1$, where SEM matches 11\% more children on average. The gap remains steady across market sizes at approximately 10\%. The standard deviations also fall sharply at any $n > 1$, confirming the general intuition that larger markets result in lower aggregate uncertainty. A similar analysis looking at the density of placement rates across market sizes shows that SEM shifts mass to the right relative to SD-RTB (Figure \ref{figure:sem-sd-density}).

\begin{figure}[t]
    \centering
    \includegraphics[width=0.7\textwidth]{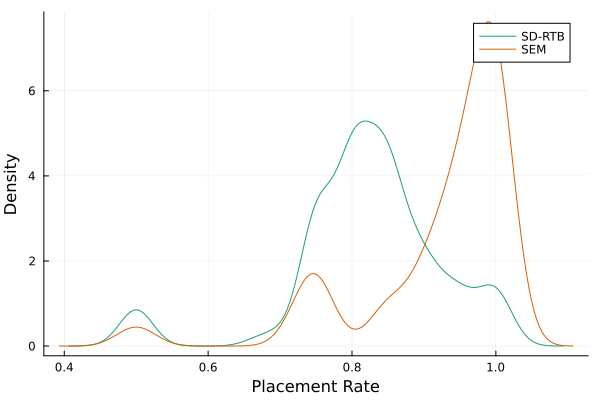}
    \caption{Density of Placement Rates}
    \label{figure:sem-sd-density}
\end{figure}

This result suggests that SEM can improve placement rates even in small markets. Intuitively, SEM's reliance on large market equilibria allows it to better internalize future arrivals and allocate objects more efficiently even with a substantial amount of aggregate uncertainty. In contrast, SD-RTB's myopic allocation can lead to suboptimal matches, especially when selective children arrive early and consume scarce resources.

\subsection*{7 \hspace{5pt} Conclusion}

In this work, we study an online matching problem where the market designer is constrained to match agents immediately upon arrival and agents have ordinal preferences that may include indifferences. We develop a novel mechanism, the Sequential Equilibrium Mechanism (SEM), that leverages large market equilibria to achieve desirable properties including asymptotic efficiency, strategyproofness, and equal-type envy-freeness. Our analysis extends to fully online matching with stochastic arrivals of both supply and demand. Preliminary simulations in a stylized setting inspired by a real-world application demonstrate that SEM can outperform existing matching processes even in small markets. 

Future work includes collaborating with our partner organization to implement SEM in practice and tightening market conditions that guarantee ordinal efficiency under the random tie-breaking schema presented in Theorem \ref{theorem:ties}. Moreover, there are open questions for equivalence results between price equilibria, random budget equilibria, and EPS: are they equivalent on the general preference domain with unit demand? Moreover, is it possible to extend existing ordinal mechanisms such as EPS or SDI to online matching? Might it be that any ordinally efficient mechanism has a large market counterpart that is asymptotically efficient in online matching? We view these as interesting questions with meaningful implications for this new problem.

%% file: Paper/Appendix.tex
\section{Proof Appendix}

Before beginning, note that we will say that a property $Z$ holds $r$-ae (almost everywhere) under measure $\mu$ for a mathematical object $O_r$ if $Z$ is true for $O_r$ for every $r \in \mathbb{R} \setminus S$ for a subset $S \subset \mathbb{R}$ such that $\mu(S) = 0$.

\begin{center}
    \textbf{Proof of Theorem \ref{theorem:existence}}.
\end{center}

\textbf{Claim 1:} $D_i(\mathbf{p})$ is non-empty and upper hemicontinuous at almost every $\mathbf{p}$. Define:
\[J_i = \{\mathbf{p} : \mathbf{p} \cdot a_i = b_i \text{ for some } a_i \in X\}\]
Since $D_i(\mathbf{p})$ is compact-valued, it is upper hemicontinuous at $p$ if it has a closed graph at $p$. Let $(\mathbf{p}^k)$ be a sequence of prices converging to $p$. $D_i(\mathbf{p})$ has a closed graph if, for all such sequences, whenever $a_i^k \in D_i(\mathbf{p}^k)$ for all $k$ with $a_i^k \rightarrow a_i$, we have:
\[a_i \in D_i(\mathbf{p})\]
Let $(\mathbf{p}^k)$ be a sequence of prices converging to some $\mathbf{p}^0 \notin J_i$. The demand for $i$ at each $\mathbf{p}^k$ is $a_i^{k} \in D_i(\mathbf{p}^k)$ with limit $a_i^{0}$. Toward a contradiction, suppose that $a_i^{0} \notin D_i(\mathbf{p}^0)$. 

Note that $a_i^{0}$ must be affordable at $\mathbf{p}^0$, otherwise $\mathbf{p}^0 \cdot a_i^{0} > b_i$. However, since $a_i^{0} \in \{0,1\}^{|X|}$ and $a_i^k \in \{0,1\}^{|X|}$ for all $k$, this implies that there exists some $K$ such that for all $j > K$, $a_i^{j} = a_i^{0}$. Then:
\[\mathbf{p}^j \cdot a_i^j = \mathbf{p}^j \cdot a_i^0\]
but
\[\mathbf{p}^m \cdot a_i^0 > b_i \text{ for large enough } m\]
so that we can take $j^* = \max\{m,j\}$, then:
\[\mathbf{p}^{j^*} \cdot a_i^{j^*} > b_i\]
contradicting $a_i^{j^*} \in D_i(\mathbf{p}^{j^*})$. So, we continue with two cases:
\begin{enumerate}
    \item There exists some $a_i \in \mathcal{A}_i$ with $\mathbf{p}^0 \cdot a_i \leq b_i$ and $a_i \succ_i a_i^{0}$.
    \item There exists some $a_i \in \mathcal{A}_i$ with $a_i \succeq_i a_i^{0}$ and $\mathbf{p}^0 \cdot a_i < \mathbf{p}^0 \cdot a_i^{0}$.
\end{enumerate}

In the first case, $\mathbf{p}^0 \cdot a_i < b_i$ since $a_i \in D_i(\mathbf{p}^0)$ and $\mathbf{p}^0 \notin J_i$. Since $\mathbf{p}^k \rightarrow \mathbf{p}^0$, there exists some $K$ such that for all $j > K$, $\mathbf{p}^j \cdot a_i < b_i$. Therefore, for all $j > K$, $a_i$ is affordable at price $\mathbf{p}^j$ and $a_i \succ_i a_i^{0} \sim_i a_i^{j}$ for large enough $j$. This contradicts that $a_i^{j} \in D_i(\mathbf{p}^j)$.

Similarly, in the second case, since $\mathbf{p}^k \rightarrow \mathbf{p}^0$, there exists some $K$ such that for all $j > K$, $\mathbf{p}^j \cdot a_i < \mathbf{p}^j \cdot a_i^{0}$. Further, for large enough $j > K$, $\mathbf{p}^j \cdot a_i^{0} = \mathbf{p}^j \cdot a_i^{j}$ and $a_i \succeq_i a_i^{0} \sim_i a_i^{j}$. This contradicts that $a_i^{j} \in D_i(\mathbf{p}^j)$.

Let $J_{i,x} = \{\mathbf{p} : \mathbf{p}_x = b_i\}$ which is singleton for each $x \in X$. (In the case of combinatorial demand, it is an affine space in $|X| - 1$, and results are unchanged.) Then:
\[J_i = \cup_{x \in X} J_{i,x}\]
Under Assumption \ref{assumption:continuous}:
\[\Pr(\mathbf{p}_x = b_i) = \Pr(\xi_x = b_i - p_x) = 0 \implies \Pr(\mathbf{p} \in J_{i,x}) = 0\]
Moreover:
\[\Pr(\mathbf{p} \in J_i) \leq \sum_{x \in X} \Pr(\mathbf{p} \in J_{i,x}) = 0 \]
where the final inequality follows because the number of events is finite. This implies that $J_i$ has measure zero. Hence, $D_i(\mathbf{p})$ is an upper hemicontinuous correspondence at almost every $\mathbf{p}$ for every $i$. Equivalently, $D_i(p + \xi)$ is upper hemicontinuous for all $p$ for $\xi$-a.e.

\textbf{Claim 2:} $L_i(p)$ and $\mathcal{D}^t(p)$ are non-empty, convex, and upper hemicontinuous correspondences in $p$. We use the following results\footnote{Each result was originally proven in \qcite{aumann1965integrals}, but his case restricts to the Lebesgue measure. We use the fully generalized versions for any non-atomic measure.} which require some definitions:

Consider any correspondence $H : I \rightrightarrows \mathbb{R}^m$ that is non-empty and closed for any $i \in I$. $H$ is \textit{integrably bounded} if there exists a single-valued, Lebesgue integrable function $h : I \rightarrow \mathbb{R}^m$ such that, for each $i \in I$ and $a \in H(i)$, $|a| \leq h(i)$. We say that $g(i)$ is a $\mu$-measurable selection of $H(i)$ if $g(i) \in H(i)$ for all $i \in I$ except for some set $I' \subset I$ such that $\mu(I') = 0$. The \qcite{aumann1965integrals} integral of $H$ with respect to an atomless probability measure $\mu$ is:
\[\int H := \bigg\{ \int_{I} g(i) d\mu(i) : g(i) \text{ is a $\mu$-measurable selection of } H(i) \bigg\}\]
\qcite{aubin-2009}. \textit{Theorem 8.1.3.} If $H$ is a measurable set-valued correspondence and is closed and non-empty for each $i \in I$, then a measurable selection of $H$ exists. This implies that a measurable selection of $D_i(\cdot)$ exists. Therefore, $L_i(p)$ is non-empty for each $p$.

\qcite{aubin-2009}. \textit{Theorem 8.6.3.} If $H$ is Borel-measurable and integrably bounded, then $\int H$ is convex-valued and compact. $D_i(p)$ is non-empty and compact-valued for each $i \in I$ and $p \in \mathbb{R}^{|X|}$. $D_i(p)$ is immediately Borel-measurable for each $i \in I$ and $p \in \mathbb{R}^{|X|}$. Any bounded, measurable function on a probability space is also integrably bounded. Therefore, $L_i(p;\xi)$ is convex-valued and compact-valued for each $p$.

Finally, upper hemicontinuity of $L_i(p)$ at each $p$ follows by the set-valued Lebesgue Dominated Convergence Theorem (\qcite{aubin-2009}. \textit{Theorem 8.6.7}) because $D_i(p + \xi)$ is upper hemicontinuous ($p + \xi$)-ae for each $p$.

Non-emptiness, convex values, and upper hemicontinuity are preserved under product with constants and summation, implying the corresponding properties for aggregate demand. With this claim established, we continue using the standard fixed-point argument.  

\textbf{Claim 3:} An equilibrium exists. There exists a price vector $\bar{p}$ with $p_x \leq K$ for each $x \in X$ such that $\mathcal{D}^t(\bar{p}) = \{(0)\}_{x \in X}$ by assumptions \ref{assumption:continuous} and \ref{assumption:budgets} (we can consider a large enough $K$ such that $K > \bar{B} - \underline{\xi}$). 

Now, we can construct the \textit{price adjustment function} (using the Minkowski sum):

\[Z(p) = \max\{p, 0\} + \bigg( \mathcal{D}^t(p) - S \bigg)\]

and define it on the domain and range such that, for all price vectors, $0 \leq p_x \leq K$ for each $x \in X$. This is a non-empty, compact, and convex set in $\mathbb{R}^{|X|}_+$, and aggregate demand is non-empty, convex, and upper hemicontinuous for all $p$. By Kakutani's Fixed Point Theorem, there exists a fixed point $p^*$ such that $Z(p^*) = p^*$. At this fixed point, we have that, for each $x \in X$:
\[p_x^* = \max\{ p_x^*, 0\} + \bigg( \mathcal{D}^t_x(p^*) - s_x \bigg)\]
Therefore, if $p^*_x > 0$, then $\mathcal{D}^t_x(p^*) = s_x$. Moreover, by construction, if $p^*_x = 0$, then $\mathcal{D}^t_x(p^*) - s_x \leq 0$ so that $\mathcal{D}^t_x(p^*) \leq s_x$. (Otherwise, this could not be a fixed point.) Therefore, $p^*$ is a price equilibrium. 

\textbf{Claim 4:} The equilibrium is greedy for some budgets $b$. 

We can write the equilibrium allocation as $\mathbf{a}^* : I \rightarrow \{0,1\}^{|X|}$ such that $\mathbf{a}^*(i) = a_i^*$ for some $F$-measurable selection $a_i^* \in L_i(p^*)$. Note then that, $i$-ae, $\mathbf{a}^*(i) \in L_i(p^*)$. Suppose that $x^* \succ_i x$ for some $x \in X$ satisfying $a_{i,x}^* > 0$. $a_{i,x}^* > 0$ implies that, for some realization $\mathbf{p}^*$:
\[\mathbf{p}^*_{x^*} > b_i \implies p^*_{x^*} > b_i - \bar{\xi}\]
otherwise, $i$ would consume $x^*$ instead of $x$. Then, we may set $b_j < b_i - \bar{\xi} + \underline{\xi}$ for all $j$ such that $t_j > t_i$. This implies that that $D_j(\mathbf{p}^*)_{x^*} = 0$ for all realizations $\mathbf{p}^*$, and therefore $a_{j,x^*}^* = 0$ for all $j$ such that $t_j > t_i$. Moreover, $\underline{B} > 0$ implies that $p^*_{x^*} > 0$ (Assumptions \ref{assumption:budgets} and \ref{assumption:bounded}), so that market clearing implies all of $x^*$ is allocated to $i'$ such that $t_{i'} \geq t_i$. This proves the Theorem. $\blacksquare$

We require two lemmas before proving Theorem \ref{theorem:fwt}.

\begin{lemma}\label{lemma:non-wasteful}
    Suppose that $p^*$ is a price equilibrium. Then if $x \succ_i D_i(\mathbf{p}^*)$ for some realization $\mathbf{p}^*$, then $p_x > 0$.
\end{lemma}

\textbf{Proof of Lemma \ref{lemma:non-wasteful}.} $x \succ_i D_i(\mathbf{p}^*)$ implies that $p_x^* + \xi > b_i$, which, when $p^*_x = 0$, is only possible if $\bar{\xi} > b_i$. Since $\bar{\xi} < \min_i b_i$, this implies that $p_x^* > 0$. $\blacksquare$

\begin{lemma}\label{lemma:strassen}
    Suppose that $\mathbf{a}^*$ is a price equilibrium allocation under $p^*$. Then:
    \begin{enumerate}
        \item If Assumption \ref{assumption:perfect-correlation} holds and $a_i$ stochastically dominates $\mathbf{a}^*(i)$ (strictly), then $p^* \cdot a_i \geq p^* \cdot \mathbf{a}^*(i)$ ($p^* \cdot a_i > p^* \cdot \mathbf{a}^*(i)$).
        \item If Assumption \ref{assumption:strong-correlation} holds, $|p_x^* - p_y^*| > 2z$ for all $x,y \in X$, and $a_i$ stochastically dominates $\mathbf{a}^*(i)$ (strictly), then $p^* \cdot a_i \geq p^* \cdot \mathbf{a}^*(i)$ ($p^* \cdot a_i > p^* \cdot \mathbf{a}^*(i)$).
    \end{enumerate}
\end{lemma}

\textbf{Proof of Lemma \ref{lemma:strassen}.} We can write $D_i(\mathbf{p}^*) \sim a_i^*$ and $A \sim a_i$. By Strassen's Theorem for FOSD, $a_i$ stochastically dominates $a_i^*$ if and only if there exists some random variable $Z$ such that $D_i(\mathbf{p}^*) = \phi_1(Z)$, $A = \phi_2(Z)$, and $\phi_2(z) \succeq_i \phi_1(z)$ for $z$-ae. 

We can embed these random variables in some probability space with events $\Omega$ with generic $\omega \in \Omega$ where $\mathbf{p}^*(\omega)$ and $Z(\omega)$ are realizations in the event $\omega$. Then, for each $\omega$:
\[\mathbf{p}^*(\omega) \cdot \phi_2(Z(\omega)) \geq \mathbf{p}^*(\omega) \cdot \phi_1(Z(\omega))\]
To see this, let $y$ be defined by $\phi_2(Z(\omega))_y = 1$ and $x$ by $\phi_1(Z(\omega))_x = 1$, if any such $x$ exists (we deal with the case that none exists shortly). The definition of demand implies that $\mathbf{p}^*_y > \mathbf{p}^*_x$ if $y \succ_i x$ by preference maximization, and $\mathbf{p}^*_y \geq \mathbf{p}^*_x$ if $y \sim x$ by expenditure minimization. Under case (1) of the Lemma:
\[\mathbf{p}^*_x = p^*_x + c \qquad\text{and}\qquad \mathbf{p}^*_y = p^*_y + c\]
implies that $p^*_y \geq p^*_x$ ($p^*_y > p^*_x$ if $y \succ_i x$). Under case (2) of the Lemma:
\[\mathbf{p}^*_x \geq p^*_x + c - z \qquad\text{and}\qquad \mathbf{p}^*_y \leq p^*_y + c + z\]
which implies:
\[p^*_y + z \geq p_x^* - z \implies p^*_y \geq p^*_x - 2z \implies p^*_y \in [p^*_x - 2z, \infty)\]
but then we know that:
\[|p^*_y - p_x^*| > 2z \implies p_y^* \notin [p^*_x - 2z, p^*_x + 2z]\]
intersecting, $p_y^* \in (p^*_x + 2z, \infty)$. So, $p^*_y > p^*_x$.

If no such $x$ exists, then $\phi_2(Z(\omega)) \succ_i D_i(\mathbf{p}(\omega))$ so that lemma \ref{lemma:non-wasteful} implies that $p_y > 0$. Hence $p^* \cdot \phi_2(Z(\omega)) > p^* \cdot \phi_1(Z(\omega)) = 0$, and, combining the cases:
\[p^* \cdot \phi_2(Z(\omega)) \geq p^* \cdot \phi_1(Z(\omega))\]
and if $\phi_2(Z(\omega)) \succ_i \phi_1(Z(\omega))$ then $p^* \cdot \phi_2(Z(\omega)) > p^* \cdot \phi_1(Z(\omega))$. Since this holds for every realization, it follows that:
\[p^* \cdot E[\phi_2(Z(\omega))] = p^* \cdot a_i \geq p^* \cdot a_i^* = p^* \cdot E[\phi_1(Z(\omega))]\]
and strictly so if $a_i$ stochastically dominates $a_i^*$ strictly because $\phi_2(Z(w)) \succ_i \phi_1(Z(w))$ must hold for some positive measure set of events by Strassen's Theorem. This proves the Lemma. $\blacksquare$

\begin{center}
    \textbf{Proof of Theorem \ref{theorem:fwt} and Theorem \ref{theorem:ties}, Part 1.}
\end{center}

Suppose that $p^*$ is a price equilibrium. Let $\mathbf{a}^* : I \rightarrow R^{|X|}$ be an equilibrium allocation such that $a_i^* \in L_i(p^*)$ and $\mathbf{a}^*(i) = a_i^*$ for all $i \in I$. Toward a contradiction, suppose that an allocation $\mathbf{a}$ stochastically dominates $\mathbf{a}^*$.

Under the assumptions for either Theorem, by Lemma \ref{lemma:strassen}, we have for all $i$:
\[p^* \cdot \mathbf{a}(i) \geq p^* \cdot \mathbf{a}^*(i)\]
and for some $i' \in I$:
\[p^* \cdot \mathbf{a}(i') > p^* \cdot \mathbf{a}^*(i')\]
Hence, we obtain:
\[\sum_{i \in I} p^* \cdot (\mathbf{a}(i) - \mathbf{a}^*(i)) f(i) = \sum_{i \neq i'} p^* \cdot (\mathbf{a}(i) - \mathbf{a}^*(i)) f(i) + p^* \cdot (\mathbf{a}(i') - \mathbf{a}^*(i')) f(i') > 0\]
However, this implies for some $x \in X$ with $p^*_x > 0$:
\[p^*_x \sum_{i \in I} (\mathbf{a}_x(i) - \mathbf{a}_x^*(i)) f(i) > 0 \implies \sum_{i \in I} \mathbf{a}_x(i) f(i) > \sum_{i \in I} \mathbf{a}_x^*(i) f(i) = s_x\]
where the last equality holds by market clearing. Therefore, $\mathbf{a}$ cannot be a feasible allocation, a contradiction. This proves the Theorem. $\blacksquare$

\begin{center}
    \textbf{Proof of Theorem \ref{theorem:ties}, Part 2.}
\end{center}

\textbf{Claim 1:} $D_i(\mathbf{p})$ is non-empty, single-valued, and continuous at almost every $\mathbf{p}$. Define:
\[J_i = \{\mathbf{p} : \mathbf{p} \cdot D_i(\mathbf{p}) = b_i\} \cup \{\mathbf{p} : \mathbf{p} \cdot a_i = \mathbf{p} \cdot a'_i \text{ for } a_i,a'_i \in X\}\]
Consider any $\mathbf{p} \notin J_i$. Then:
\[D_i(\mathbf{p}) = \{a_i\} \text{ for some } a_i \in X\]
and $\mathbf{p} \cdot a_i < b_i$. Let $R_i(a_i) = \{a'_i \in X : a'_i \succ_i a_i\}$ and $\bar{R}_i(a_i) = \{a'_i \in X : a'_i \succeq_i a_i\}$. For any $a'_i \in R_i(a_i)$, $\mathbf{p} \cdot a'_i > b_i$ by assumption that $D_i(\mathbf{p}) = \{a_i\}$. Similarly, $a'_i \in \bar{R}_i(a_i)$ implies $\mathbf{p} \cdot a'_i > \mathbf{p} \cdot a_i$.

There exists some $\epsilon > 0$ such that for any $\mathbf{p}'$ where $\|\mathbf{p} - \mathbf{p}'\|_\infty < \epsilon$:
\[D_i(\mathbf{p}') = \{a_i\}\]
We can consider $\epsilon < b_i - \mathbf{p} \cdot a'_i$ for all $a'_i \in R_i(a_i)$ so that $\mathbf{p}' \cdot a'_i > b_i$ for any $a'_i \in X$ such that $a'_i \succ_i a_i$. Alternatively, we may consider $\epsilon < \mathbf{p} \cdot a_i - \mathbf{p} \cdot a'_i$ for all $a'_i \in \bar{R}(a_i)$ so that $\mathbf{p}' \cdot a'_i > \mathbf{p}' \cdot a_i$ for any $a'_i \in X$ such that $a'_i \succeq_i a_i$.

This implies that $D_i(\cdot)$ is single-valued and locally constant at any $\mathbf{p} \notin J_i$; it is therefore non-empty and continuous at every $\mathbf{p} \notin J_i$. 

As in Theorem \ref{theorem:existence}, let $J_{i,x} = \{\mathbf{p} : \mathbf{p}_x = b_i\}$. Let $J_{i,(x,y)} = \{\mathbf{p} : \mathbf{p}_x = \mathbf{p}_y\}$. We have already shown that $\Pr(\mathbf{p} \in J_{i,x}) = 0$. Moreover:
\[\Pr(\mathbf{p} \in J_{i,(x,y)}) = \Pr(p_x + c + \zeta_x = p_y + c + \zeta_y) = \Pr(\zeta_x - \zeta_y = p_x - p_y) = 0\]
where the last equality follows since, by Assumptions \ref{assumption:continuous} and \ref{assumption:strong-correlation}, $\zeta_x - \zeta_y$ is a continuously distributed random variable. Then:
\[\Pr(\mathbf{p} \in J_i) \leq \sum_{x \in X} \bigg[ \Pr(\mathbf{p} \in J_{i,x}) + \sum_{y \in X} \Pr(\mathbf{p} \in J_{i,(x,y)}) \bigg] = 0\]
where the last inequality follows since this is a finite sum of measure zero events. Therefore, $J_i$ has zero measure. The Claim follows.

\textbf{Claim 2:} $L_i(p)$ is a non-empty and continuous function in $p$. By Claim 1, we can simply use the Lebesgue integral instead of Aumann integration to obtain a function that is non-empty and continuous. These properties are preserved under product and summation, implying that they hold for aggregate demand. $\blacksquare$

We establish a useful Lemma before proving Theorem \ref{theorem:sem}.
\begin{lemma}\label{lemma:continuation}
    Suppose that $p^t$ is a $(t)$-equilibrium and $\mathbf{a}^t$ is a $(t)$-equilibrium allocation for $p^t$. Then $p^t$ is a $(t+1)$-price equilibria under supply $S - \sum_{i \in I} \mathbf{a}^t(i) f^t(i)$, and $\mathbf{a}^t$ is a $(t+1)$-equilibrium allocation.
\end{lemma}

Let the conditions for the Lemma hold. We have that:
\[\sum_{i \in I} \mathbf{a}^t(i) f^{t+1,T}(i) \in \mathcal{D}^{t+1}(p^t)\]
Moreover:
\[\sum_{i \in I} \mathbf{a}^t(i) f^{t,T}(i) = \sum_{i \in I} \mathbf{a}^t(i) f^t(i) + \sum_{i \in I} \mathbf{a}^t(i) f^{t+1,T}(i) \leq S\]
implies
\[\sum_{i \in I} \mathbf{a}^t(i) f^{t+1,T} \leq S - \sum_{i \in I} \mathbf{a}^t(i) f^t(i)\]
Clearly, this must also hold with equality for any $x \in X$ such that $p_x > 0$. This proves the Lemma. $\blacksquare$

\begin{center}
    \textbf{Proof of Theorem \ref{theorem:sem}}.
\end{center}

    
Before beginning, we define some notation. We use the $L_\infty$ norm to measure distance. That is:
\[\|a(i) - a'(i)\|_{\infty} = \sup_{x \in X} \|a(i,x) - a'(i,x)\|\]
Then:
\[\|a(\cdot) - a'(\cdot)\|_{\infty} = \sup_{i \in I} \|a(i) - a'(i)\|_{\infty}\]
The underlying space includes distributions over $I$, supply vectors, price vectors, and allocations. We represent this as:
\[\mathcal{C} = \mathcal{P} \times \mathbb{R}^{|X|} \times \mathbb{R}^{|X|} \times \mathcal{A}\]
Finally, we can use the $L_\infty$ product metric (we write this as $\|\cdot\|^*_\infty$) to define a metric space on $\mathcal{C}$. An element in this product metric space converges to another element if and only if each component converges in the respective metric.

Suppose that $(\theta_n) = (\theta_1, \theta_2, ..., \theta_n)$ is a sequence of i.i.d random vectors with expectation $E[\theta]$. We say that $\theta_n \prightarrow \theta$ (uniform convergence in probability) if:
\[\lim_{n \rightarrow \infty} \Pr\bigg( \|\theta_n - \theta\|_\infty > \epsilon\bigg) = 0\]
Now, by assumption of $\ell$-regularity (as argued in the text, this is implied by regularity under RTB), we can select an equilibrium $\mathbf{a}^* \in \mathcal{E}^1(\nu)$ such that $\mathbf{a}^* = \lambda_n(\nu) := \lambda(\nu)$ for some continuous function $\lambda$. Fix this equilibrium for the remainder of the proof. 

\textbf{Claim 1:} We will show that $\tilde{F}^1$ approximates $F^1$, and, therefore, $S^2/n$ also approximates $S - \sum_{i \in I} \mathbf{a}^*(i) f^1(i)$.

Let $\tilde{F}_n^1(i)$ denote the empirical distribution in period $1$ for an $n$-fold replica, that is:
\[\tilde{F}_n^1(i) = \frac{1}{n} \sum_{i' \in \tilde{I}^1(n)} \mathbbm{1} \{i' \leq i\}\]
Analogously, $\tilde{F}_n(\cdot;1) = (\tilde{F}_n^1(\cdot), F^2(\cdot), ..., F^T(\cdot))$. Each $i' \in \tilde{I}^1(n)$ can be viewed as a random variable which is an arrival independently sampled from the conditional distribution $F^1$. By the Glivenko–Cantelli Theorem, we have:
\[\tilde{F}_n^1 \prightarrow F^1\]
The equilibrium correspondence is homogenous degree one. This implies that an equilibrium for the scaled empirical distribution $\nu_n^1 = (\tilde{F}(\cdot, 1), nS^1/n)$ is also an equilibrium for the replica market: 
\[\alpha^1_n \in E^1(\nu_n^1) \implies \alpha^1_n \in E^1(n\tilde{F}(\cdot;1), nS)\]
$S^1 = S$ implies that $\nu^1_n \prightarrow \nu$. Hence, for any $\epsilon > 0$, there exists a large enough $n$ such that, for any $k > 0$:
\[\Pr\bigg(\|\nu^1_n - \nu\|_\infty^* > \epsilon\bigg) < k\]
By $\ell$-regularity, for any $\delta > 0$, we can thus find an $n$ such that we can select $\alpha^1_n$ satisfying:
\[\Pr\bigg(\| \alpha^1_n - \mathbf{a}^* \|_\infty > \delta \bigg) < k\]
so $\alpha_n^1 \prightarrow \mathbf{a}^*$. Thus:
\begin{align*}
    &\sum_{i \in I} \big( \alpha^1_n(i) - \mathbf{a}^*(i) \big) \tilde{f}_n^1(i) \leq \sup_{i \in \tilde{I}^1(n)} \big( \alpha_n^1(i) - \mathbf{a}^*(i) \big) \prightarrow \{0\}^{|X|}\\
    &\implies \frac{1}{n} \sum_{\tilde{i} \in \tilde{I}^1(n)} \alpha^1_n(\tilde{i}) \prightarrow \frac{1}{n} \sum_{\tilde{i} \in \tilde{I}^1(n)} \mathbf{a}^*(\tilde{i}) \prightarrow \sum_{i \in I} \mathbf{a}^*(i) f^1(i)
\end{align*}
The first convergence follows by assumption that $\|\alpha_n^1 - \mathbf{a}^*\|_\infty \prightarrow 0$. The last convergence follows because each $\mathbf{a}^*(\tilde{i})$ is an independent, identically distributed random variable. By the WLLN, the sample average converges to its expectation for any distance norm.

By \qcite{gandhi2002dependent}, there exists a dependent rounding scheme producing a feasible ex-post allocation (with probability one) such that, letting $\tilde{\alpha}^1_n$ be a draw from a distribution $G_n$ over ex-post allocations, (i) $E_{G_n}[\tilde{\alpha}^1_n(i')] = \alpha^1_n(i')$ and (ii) for any $x \in X$, $Cov(\tilde{\alpha}^1_{n,x}(i), \tilde{\alpha}^1_{n,x}(j)) \leq 0$ for every $i,j \in \tilde{I}^1(n)$. Note then that $\tilde{\alpha}^1_n$ is conditional upon observing arrivals and is random only with respect to randomization over ex-post feasible allocations. By (ii):
\begin{align*}
    Var\bigg(\frac{1}{n} \sum_{i' \in \tilde{I}^1(n)} \tilde{\alpha}^1_{n,x}(i')\bigg) &= \frac{1}{n^2} \sum_{i' \in \tilde{I}^1(n)} Var(\tilde{\alpha}^1_{n,x}(i')) + \frac{1}{n^2} \sum_{i' \neq j'} Cov(\tilde{\alpha}^1_{n,x}(i'), \tilde{\alpha}^1_{n,x}(j'))\notag\\
    &\leq \frac{1}{n^2} \sum_{i' \in \tilde{I}^1(n)} Var(\tilde{\alpha}^1_{n,x}(i'))\notag\\
    &\leq \frac{1}{4n}\label{eq:var-bound}
\end{align*}
where the final inequality follows since each $\tilde{\alpha}^1_{n,x}(i') \in \{0,1\}$. Applying Chebyshev's inequality for each $x \in X$ and a union bound:
\[\Pr\Bigg( \bigg\|\frac{1}{n} \sum_{i' \in \tilde{I}^1(n)} \tilde{\alpha}^1_n(i') - \frac{1}{n} \sum_{i' \in \tilde{I}^1(n)} \alpha^1_n(i')\bigg\|_\infty > \delta \Bigg) \leq \frac{|X|}{4n\delta^2} \ninfrightarrow 0\]
therefore:
\[\frac{1}{n} \sum_{i' \in \tilde{I}^1(n)} \tilde{\alpha}^1_n(i') \prightarrow \frac{1}{n} \sum_{i' \in \tilde{I}^1(n)} \alpha^1_n(i') \implies \frac{1}{n} \sum_{i' \in \tilde{I}^1(n)} \tilde{\alpha}^1_n(i') \prightarrow \sum_{i \in I} \mathbf{a}^*(i) f^1(i)\]
Moreover, $t = 2$ supply is:
\[S^2 = nS - \sum_{i' \in \tilde{I}^1(n)} \tilde{\alpha}^1_n(i') \implies \frac{S^2}{n} \prightarrow S - \sum_{i \in I} \mathbf{a}^*(i) f^1(i)\]
\textbf{Claim 2:} For each $t \in \{1, 2, ..., T-1\}$, $S^{t+1}/n \prightarrow S - \sum_{i \in I} \mathbf{a}^*(i) f^{1,t}(i)$. Claim 1 proved the base case for $t = 1$. We assume that the claim holds for $t$ and prove that it holds for $t + 1$.

By the same logic as Claim 1:
\[\tilde{F}^{t+1}_n \prightarrow F^{t+1}\]
Define:
\[\nu_n^{t+1} = \bigg(\tilde{F}(\cdot;t+1), \frac{S^{t+1}}{n}\bigg)\]
The induction hypothesis implies that $\nu_n^{t+1} \prightarrow (F, S - \sum_{i \in I} \mathbf{a}^*(i) f^{1,t}(i))$. By Lemma \ref{lemma:continuation}, $\mathbf{a}^* \in \mathcal{E}^{t+1}(F, S - \sum_{i \in I} \mathbf{a}^*(i) f^{1,t}(i))$. 

Similarly, $\ell$-regularity implies that, for large enough $n$, we can select $\alpha^{t+1}_n \in \mathcal{E}^{t+1}(\nu^{t+1}_n)$ such that $\alpha^{t+1}_n \prightarrow \mathbf{a}^*$ to obtain:
\[\frac{1}{n} \sum_{i' \in \tilde{I}^{t+1}(n)} \tilde{\alpha}^{t+1}_n(i') \prightarrow \sum_{i \in I} \mathbf{a}^*(i) f^{t+1}(i)\]
and conclude that:
\[\frac{S^{t+2}}{n} \prightarrow \frac{S^{t+1}}{n} - \sum_{i \in I} \mathbf{a}^*(i) f^{t+1}(i) = S - \sum_{i \in I} \mathbf{a}^*(i) f^{1,t+1}(i)\]
\textbf{Claim 3:} The SEM allocation $\alpha_n$ is asymptotically efficient. By Claims 1 and 2, we can conclude that $\alpha^t_n \prightarrow \mathbf{a}^{*}$ for each $t$. Therefore, $\alpha_n \prightarrow \mathbf{a}^*$. $\mathbf{a}^*$ is ordinally efficient by either Theorem \ref{theorem:fwt} or \ref{theorem:ties}. This proves the first part of the Theorem. 

\textbf{Claim 4:} $\alpha_n$ is equal-type envy-free for any $n$. Toward a contradiction, suppose that $\alpha^k_n(j) \succ_i \alpha^t_n(i)$ for some $i,j \in \tilde{I}$ with $k > t$. This implies there exists some $x$ such that $s_x^k > 0$ and $x \succ_i y$ for some $y$ with $\alpha^t_{n,y}(i) > 0$. However:
\[s_x^k > 0 \implies s_x^t - \sum_{i \in I} \alpha_n^t(i) f^t(i) > 0\]
Then, $x \succ_i y$ implies that $\mathbf{p}^t_x > b_i$ for some realization $\mathbf{p}^t$. By Theorem \ref{theorem:fairness}, $p_x^t - \xi > b_j$ for any $j$ with $t_j > t_i$. So:
\[s_x^t - \sum_{i \in I} \alpha_n^t(i) f^{t,T}(i) = s_x^t - \sum_{i \in I} \alpha_n^t(i) f^t(i) > 0\]
By market clearing, we must have that $p^t_x = 0$. This contradicts $\alpha^t_{n,y}(i) > 0$. Moreover, by Theorem \ref{theorem:fairness}, the allocation is envy-free within time periods. Hence, $\alpha_n$ is equal-type envy-free for each $n$. 

\textbf{Claim 5:} SEM is strategyproof. This immediately follows because by declaring truthfully, an agent arriving at $t$ has the largest budget among remaining agents. She must receive a stochastically undominated lottery given remaining supply at $t$ by Theorem \ref{theorem:existence}. 

\textbf{Claim 6:} SEM is ex-post greedy for each $n$. Let $\tilde{\alpha}^t_n$ be some matching drawn from the distribution of $\alpha^t_n$. Toward a contradiction, suppose that there exists some $x$ such that $x \succ_i y$ for some $y$ with $\alpha^t_{n,y}(i) > 0$ and $\tilde{\alpha}^k_{n,x}(j) > 0$ for some $k > t$. This implies that $s_x^k > 0$. We obtain the same contradiction as in Claim 4. Hence, $\tilde{\alpha}_n$ is ex-post greedy for each $n$. This proves the Theorem. 

\textbf{Corollary 1:} Fix an agent $i$, a misreport $j$ with $t_j=t_i$, and $\varepsilon>0$.   Write
\[
\pi_{n,i} := \pi_{n,i}^T(\tilde I^{1,T}(n))
\qquad\text{and}\qquad
\pi_{n,(i,j)} := \pi_{n,i}^T(\{j\}\cup \tilde I^{1,T}(n) \setminus\{i\})
\]
for agent $i$'s allocation under truthful reporting and under the misreport, respectively.

By Theorem \ref{theorem:sem}, there exists an equilibrium allocation $\mathbf{a}^*$ with $\pi_i := \mathbf{a}^*(i)$ and $\pi_j := \mathbf{a}^*(j)$ such that
\[
\pi_{n,i} \prightarrow \pi_i
\qquad\text{and}\qquad
\pi_{n,(i,j)} \prightarrow \pi_j .
\]
In particular, such an equilibrium must exist because a single agent's misrepresentation affects the empirical distribution by $1/n$ at most, so the empirical distribution under either truth-telling or misrepresentation converges to the same limit. The proof of Theorem \ref{theorem:sem} follows to show that the limiting equilibrium allocation is equivalent. By the intersection bound,
\[
\lim_{n \rightarrow \infty}\Pr\!\left(
\|\pi_{n,i}-\pi_i\|_\infty < \delta
\;\land\;
\|\pi_{n,(i,j)}-\pi_j\|_\infty < \delta
\right)
= 1
\quad\text{for all }\delta>0 .
\]

We distinguish two cases.

\noindent\emph{Case 1: $\pi_i=\pi_j$.}
For any $\epsilon > 0$, choose $\delta\le \varepsilon/2$. On the event above,
\[
\|\pi_{n,i}-\pi_{n,(i,j)}\|_\infty
\le
\|\pi_{n,i}-\pi_i\|_\infty
+
\|\pi_{n,(i,j)}-\pi_j\|_\infty
< 2\delta \le \varepsilon .
\]
Hence,
\[
\lim_{n \rightarrow \infty} \Pr\!\left(
\|\pi_{n,i}-\pi_{n,(i,j)}\|_\infty < \varepsilon
\right)
= 1 
\]

\medskip
\noindent\emph{Case 2: $\pi_i\neq\pi_j$.}
Lemma \ref{lemma:strassen} and $b_i = b_j$ imply that it is not true that $\pi_j \succeq_i \pi_i$. Thus there exists a cutoff $x$ such that:
\[
\sum_{y\succeq_i x} \pi_{j,y}
<
\sum_{y\succeq_i x} \pi_{i,y} .
\]
Let
\[
\gamma
:=
\sum_{y\succeq_i x} \pi_{i,y}
-
\sum_{y\succeq_i x} \pi_{j,y}
>0 .
\]
For any $\epsilon > 0$, choose $\delta < \min\{\varepsilon/2,\gamma/3\}$.
On the event that both $\pi_{n,i}$ and $\pi_{n,(i,j)}$ are within $\delta$ of their limits, we have
\[
\sum_{y\succeq_i x} \pi_{n,(i,j),y}
\le
\sum_{y\succeq_i x} \pi_{j,y} + \delta
<
\sum_{y\succeq_i x} \pi_{i,y} - \gamma + \delta
\le
\sum_{y\succeq_i x} \pi_{n,i,y} - (\gamma - 2\delta),
\]
where $-(\gamma - 2\delta)$ is strictly negative by construction.
Thus $\pi_{n,(i,j)}$ does not stochastically dominate $\pi_{n,i}$. Therefore,
\[
\lim_{n \rightarrow \infty} \Pr\!\left(
\pi_{n,i} \succeq_i \pi_{n,(i,j)}
\right)
= 1
\]
Finally:
\begin{align*}
    \lim_{n \rightarrow \infty} \Pr\!\big(
    \underbrace{\pi_{n,i} \succeq_i \pi_{n,(i,j)}}_{E_1} 
    \; \lor \; 
    \underbrace{\| \pi_{n,i} - \pi_{n,(i,j)} \|_{\infty} < \epsilon}_{E_2} 
\big) 
&\geq
\lim_{n \rightarrow \infty} \bigg( 
    \max\{
        \Pr\!\left(
            E_1
        \right)
        ,
        \Pr\!\left(
            E_2
        \right)
    \}
\bigg) \\
&= 
\max\bigg\{ 
    \lim_{n \rightarrow \infty} \Pr\!\left(
            E_1
        \right)
    ,
    \lim_{n \rightarrow \infty} \Pr\!\left(
            E_2
        \right)
\bigg\}\\
&= 1
\end{align*}
This proves the Corollary. $\blacksquare$

\begin{center}
    \textbf{Proof of Theorem \ref{theorem:fully-online}}.
\end{center}

By Theorem \ref{theorem:existence}, we know that lottery demand satisfies the necessary properties for Kakutani's Fixed Point Theorem. Hence, we can begin from this point.

\textbf{Claim 1:} A $(t)$-Lindahl equilibrium exists. Assume there exists prices such that:
\begin{align*}
    \mathcal{D}^{t,k-1}(p^t, p^{t+1}, ..., p^{k-1}) \leq \mathcal{S}^{t,k-1}\tag{Induction}
\end{align*}
By Theorem \ref{theorem:existence}, we know such prices exist for $k = t$. We skip this base case. 

There exists a price vector $\bar{p}^k$ with $p_x^k \leq K$ for each $x \in X$ such that $\mathcal{D}^{k,k}(\bar{p}^k) = \{(0)\}_{x \in X}$ by assumptions \ref{assumption:continuous} and \ref{assumption:budgets} (we can consider a large enough $K$ such that $K > 1 - \underline{\xi}$). This, in turn with the inductive hypothesis, implies:
\[\mathcal{D}^{t,k}(p^t, p^{t+1}, ..., p^{k-1}, \bar{p}^k) = \mathcal{D}^{t,k-1}(p^t, p^{t+1}, ..., p^{k-1}) \leq S^{t,k-1} \leq S^{t,k}\]

The \textit{price adjustment function} for period-$k$ prices (using the Minkowski sum) is:

\[Z^k(p) = \max\{p^k, 0\} + \bigg( \mathcal{D}^{t,k}(p) - \mathcal{S}^{t,k} \bigg)\]

and define it on the domain and range such that, for all price vectors, $0 \leq p_x^t \leq K$ for each $x \in X$. By Kakutani's Fixed Point Theorem,  there exists a fixed point $p^*$ such that $Z^k(p^*) = p^{*,k}$ for each $k \in [T]$ such that $k \leq t$. At this fixed point, we have that, for each $x \in X$:
\[p_x^{*,k} = \max\{ p_x^{*,k}, 0\} + \bigg( \mathcal{D}^{t,k}_x(p^*) - \mathcal{S}^{t,k}_x \bigg)\]
Therefore, if $p^{*,k}_x > 0$, then $\mathcal{D}^{t,k}_x(p^*) = \mathcal{S}^{t,k}_x$. Moreover, by construction, if $p^{*,k}_x = 0$, then $\mathcal{D}^{t,k}_x(p^*) \leq \mathcal{S}_x^{t,k}$. Applying induction suffices to find such equilibrium prices for every $k \in [T]$ such that $k \geq t$. Therefore, we can construct $p^*$ which will be a $(t)$-Lindahl equilibrium.

\textbf{Claim 2:} The equilibrium allocation is greedy.

We can write the equilibrium allocation as $\mathbf{a}^* : I \rightarrow \{0,1\}^{|X|}$ such that $\mathbf{a}^*(i) = a_i^*$ for some $F$-measurable selection $a_i^* \in L_i(p^*)$. Suppose that $x^* \succ_i x$ for some $x \in X$ satisfying $a_{i,x}^* > 0$. Let $k := t_i$. $a_{i,x}^* > 0$ implies that, for some realization $\mathbf{p}^*$:
\[\mathbf{p}^{*,k}_{x^*} > 1 \implies p^{*,k}_{x^*} > 1 - \bar{\xi} > 0\]
where the right-hand implication follows by Assumption \ref{assumption:bounded}. Market clearing implies that $\mathcal{D}_x^{t,k}(p^*) = \mathcal{S}^{t,k}_x$. This proves the Theorem. $\blacksquare$

%% file: Paper/Appendix-B.tex
\section{Simulation Appendix}

Here, we provide an additional simulation to support the claim that $\ell$-regularity holds generically. We simulate one-hundred instances of large markets with three objects and thirteen agent types which permute all preference rankings over the three objects (we exclude the outside option for convenience; this ensures that positive prices occur semi-frequently). We set $T = 3$ to match aggregate supply with aggregate demand. The base arrival rates are drawn from a Dirichlet distribution with identical parameters across coordinates. Budgets and prices are set to satisfy Assumptions \ref{assumption:continuous} through \ref{assumption:perfect-correlation}.

For each market instance, we simulate one-hundred perturbations of the market which adds a Dirichlet distributed shock vector scaled by $\epsilon$ to the arrival densities; we then re-normalize them to sum to one and compute the equilibrium prices. Our algorithm computes market clearing equilibria within a tolerance of $0.01$, so we set $\epsilon = 0.025$ to allow prices to change from the base market.

In our results below, we report the infinity norm in the prices between each perturbed market and the base market. We report this distance metric as normalized by the median budget to allow for meaningful comparisons. Moreover, we also report whether or not prices that are tied in the base market remain tied in the perturbed market.

\begin{table}[t]
    \centering
    \caption{Perturbation Results}
    \begin{tabular} {l@{\hskip 7ex}c}
    \hline
    \hline
    \textit{Metric} & \\ \hline
    \textit{Average Price Distance} & 0.012 \\
    \textit{Average \% Preserved Ties} & 100\% \\
    \textit{Average Clearing Error} & 0.007 \\ \hline
    \end{tabular}
    
    \label{table:perturbations}
\end{table}

Table \ref{table:perturbations} summarizes our results. On average, prices change very little relative to budgets for small perturbations, indicating that most markets are regular. Moreover, all ties in the base market remain tied in the perturbed markets, suggesting that $\ell$-regularity also holds. Finally, the average clearing error remains below our tolerance threshold, demonstrating that our algorithm is generally able to converge to an equilibrium.